\documentclass[11pt]{article}
\usepackage{amsfonts}
\usepackage{epsfig}
\usepackage{amssymb}
\usepackage{amstext}
\usepackage{amsmath}
\usepackage{theorem}
\usepackage{algorithm}
\usepackage{algpseudocode}
\usepackage[left=1in,right=1in,top=1in,bottom=1in,nohead]{geometry}
\usepackage[active]{srcltx}
\allowdisplaybreaks[2]
\def\mstar{\mnote{****}}

\newcommand{\half}{\frac{1}{2}}

\newcommand{\card}[1]{\left|#1\right|}
\newenvironment{proof}{{\bf Proof:}}{\hfill\mbox{$\Box$}}
\newenvironment{proofof}[1]{{\bf Proof of #1:  }}{\hfill\mbox{$\Box$}}

\def\bx{{\bf x}}
\def\by{{\bf y}}
\def\br{{\bf r}}
\def\bs{{\bf s}}
\def\bt{{\bf t}}

\def\cA{{\cal A}}

\def\cP{{\cal P}}
\def\cE{{\cal E}}
\def\bx{{\bf x}}
\def\a{\alpha}
\def\b{\beta}
\def\d{\delta}
\def\D{\Delta}
\def\e{\varepsilon}
\def\f{\phi}
\def\g{\gamma}
\def\G{\Gamma}

\def\z{\zeta}
\def\th{\theta}

\def\m{\mu}

\def\p{\pi}

\def\r{\rho}

\def\s{\sigma}
\def\S{\Sigma}
\def\t{\tau}

\def\whp{{\bf whp}}

\newcommand{\mnote}[1]{\marginpar{\footnotesize\raggedright#1}}
\newcommand{\set}[1]{\left\{#1\right\}}

\newcommand{\proofend}{\hspace*{\fill}\mbox{$\Box$}}

\def\cX{{\cal X}}

\def\E{{\sf E}}
\def\Pr{{\sf P}}

\def\cC{{\cal C}}
\def\cB{{\cal B}}

\def\cG{{\cal G}}
\newcommand{\ignore}[1]{}

\def\Gv{\cX_{\vv}}
\def\Poi{\mathop{\mathrm{Poi}}}
\def\vv{\vec{v}}
\def\yy{\vec{y}}
\def\cQ{{\cal Q}}
\newtheorem{theorem}{Theorem}

\newtheorem{lemma}[theorem]{Lemma}
\newtheorem{corollary}[theorem]{Corollary}

\newtheorem{remark}[theorem]{Remark}
\newtheorem{claim}[theorem]{Claim}

\newcommand{\brac}[1]{\left(#1\right)}
\newcommand{\bfrac}[2]{\brac{\frac{#1}{#2}}}

\newcommand{\beq}[1]{\begin{equation}\label{#1}}
\newcommand{\eeq}{\end{equation}}
\newcommand{\blem}[1]{\begin{lemma}\label{#1}}
\newcommand{\elem}{\end{lemma}}
\newcommand{\bth}[1]{\begin{theorem}\label{#1}}
\newcommand{\enth}{\end{theorem}}
\newcommand{\brem}[1]{\begin{remark}\label{#1}}
\newcommand{\erem}{\end{remark}}

\def\c{{3/4}}
\def\go{{3/5}}

\author{Alan Frieze\thanks{Supported in part by NSF Grant DMS0753472.}, P\'all Melsted\\
Department of Mathematical Sciences\\Carnegie Mellon University\\Pittsburgh PA 15213\\
U.S.A.}
\title{Maximum Matchings in Random Bipartite Graphs and the Space Utilization of Cuckoo Hashtables}

\date{}

\begin{document}
\maketitle

\begin{abstract}
We study the the following question in Random Graphs. 
We are given two disjoint sets $L,R$ with $|L|=n=\a m$ and $|R|=m$. We construct a random graph $G$ by allowing each
$x\in L$ to choose $d$ random neighbours in $R$. The question discussed is as to the size $\m(G)$ of the largest matching in $G$.
When considered in the context of Cuckoo Hashing, one key question is as to when is $\m(G)=n$ \whp? 
We answer this question exactly when $d$
is at least three. We also establish a precise threshold for when Phase 1 of the Karp-Sipser Greedy matching algorithm
suffices to compute a maximum matching \whp.
\end{abstract}

\section{Introduction}
For a graph $G$ we let $\m(G)$ denote the size of the maximum matching in $G$.
In essence this paper provides an analysis of $\m(G)$ 
in the following model of a random bipartite graph. We have two disjoint
sets $L,R$ where $L=[n],R=[m]$ where $n=\a m$. Each $v\in L$ independently chooses $d$ random vertices of $R$ as neighbours. 
Our assumptions are that $\a>0,\,d\geq 3$ are fixed and $n\to \infty$. One motivation for this study comes from
Cuckoo Hashing.

Briefly each one of $n$ items $x\in L$ has
$d$ possible locations $h_1(x),h_2(x),\ldots,h_d(x)\in R$, where $d$ is
typically a small constant and the $h_i$ are hash functions, typically
assumed to behave as independent fully random hash functions.  (See
\cite{MVad} for some justification of this assumption.)  
We are thus led to consider the bipartite graph $G$ which has vertex set $L\cup R$ and edge set $\set{(x,h_j(x)):\;x\in L,j=1,,2,\ldots,d}$. 
Under the assumption that the hash functions are completely random we see that $G$ has the 
same distribution as the random graph defined in the previous paragraph. 

We assume
each location can hold only one item.  When an item $x$ is inserted
into the table, it can be placed immediately if one of its $d$ locations is
currently empty.  If not, one of the items in its $d$ locations must
be displaced and moved to another of its $d$ choices to make room
for $x$.  This item in turn may need to displace another item out of one
its $d$ locations.  Inserting an item may require a sequence of moves,
each maintaining the invariant that each item remains in one of its
$d$ potential locations, until no further evictions are needed. 
Thus having inserted $k$ items, we have constructed a matching $M$ of size $k$ in $G$. Adding a $(k+1)$'th
item is tantamount to constructing an augmenting path with repsect to $M$.
All $n$ items will be insertable
in this way iff $G$ contains a matching of size $n$.

The case of $d=2$ choices is notably different from that for other values
of $d$ and the theory for the case where there are $d = 2$ 
bucket choices for each item is well understood at this point
\cite{DM,kutz,cuckoo}. We will therefore assume that $d\geq 3$.

We will now revert to the abstract question posed in first paragraph of the paper.
\section{Definitions and Results}
This question was studied to some extent by Fotakis, Pagh, Sanders and Spirakis \cite{fotakis}.
They show in the course of their analysis of Cuckoo hashing that the following holds:
\begin{lemma}\label{fot}
Suppose that $0<\e<1$ and $d\geq 2(1+\e)\log(e/\e)$. Suppose also that $m=(1+\e)n$. Then
\whp\ $G$ contains a matching of size $n$ i.e. a matching of $L$ into $R$.
\end{lemma}
\proofend

In particular, if $d=3$ and $m\approx 1.57n$ then Lemma \ref{fot} shows that there is a matching of $L$ into $R$ \whp.

This lemma is not tight and recently  Mitzenmacher et al \cite{DMP} observed a connection with a result of
Dubois and Mandler on Random 3-XORSAT \cite{DuMa} that enables one to essentially answer the question as to when
$\m(G)\geq n$ for the case $d=3$. More recently, Fountoulakis and Panagiotou \cite{FP} have established 
thresholds for when there is a matching of $L$ into $R$ \whp, for all $d\geq 3$.

We begin with a simple observation that is the basis of the Karp-Sipser Algorithm
\cite{KS,AFP}. If $v$ is a vertex of degree one in $G$ and $e$ is its unique incident edge,
then there exists a maximum matching of $G$ that includes $e$. Karp and Sipser exploited this via a simple 
greedy algorithm:
 \begin{algorithm}
\caption{Karp-Sipser Algorithm}\label{KS}
\begin{algorithmic}[1]
\Procedure{KSGreedy}{$G$}
\State $M \gets \varnothing $
\While{$\G \neq \varnothing $}
  \If {$\G$ has vertices of degree one}
    \State Select a vertex $\xi$ uniformly at random from the set of vertices of degree one
    \State Let $e=(\xi,\eta)$ be the edge incident to $\xi$
  \Else
    \State Select an edge $e=(v,u)$ uniformly at random\label{phase2}
  \EndIf
  \State $M\gets M \cup \{e\}$
  \State $\G\gets \G\setminus\{\xi,\eta\}$ 
\EndWhile
\State \textbf{return} M
\EndProcedure
\end{algorithmic}
\end{algorithm}

Phase 1 of the Karp-Sipser Algorithm ends and Phase 2 begins when the graph remaining has minimum degree at least two. 
So if $\G_1$ denotes the graph $\G$ remaining at the end of
Phase 1 and $\t_1$ is the number of iterations involved in Phase 1 then
\beq{simple}
\m(G)=\t_1+\m(\G_1).
\eeq
Our approach to estimating $\m(G)$ is to (i) obtain an asymptotic expression for $\t_1$ that holds \whp\ and then (ii) show that 
\whp\ $\G_1$ has a (near) perfect matching and then apply \eqref{simple}.

We summarise our results as follows: Let $z_1$ satisfy 
\beq{zaa}
z_1=\frac{e^{z_1}-1}{d-1}
\eeq
and let
\beq{azaa}
\a_1=\frac{z_1}{d(1-e^{-z_1})^{d-1}}.
\eeq
\begin{theorem}\label{th1}\ 
If $\a\leq\a_1$ then \whp\ $\m(G)=\t_1=n$.
\end{theorem}
Thus \whp\ Phase 1 of the Karp-Sipser Algorithm finds a (near) maximum matching if $\a\leq \a_1$.
In particular, if $d=3$ then $z_1\approx 1.251$ and $\a_1\approx .818$ and thus $m\approx 1.222n$ is enough for a
matching of $L$ into $R$.

Andrea Montanari has pointed out that our proof of Theorem \ref{th1} via the differential equations method is not
new and already appears in Luby, Mitzenmacher, Shokrollahi and Spielman \cite{LMSS} and also in Dembo and Montanari \cite{DeMo}.
We will prune this from the final version of the paper, but leave it in here for now.

Now consider larger $\a$. Let $z^*$ be the largest non-negative solution to 
$$\bfrac{z}{\a d}^{\frac{1}{d-1}} + e^{-z}-1=0.$$
\begin{theorem}\label{th2}
If $\a>\a_1$ then \whp
\begin{description}
\item[(a)] $z^*>0$.
\item[(b)] $\t_1\sim n\brac{1-\bfrac{z^*}{\a d}^{\frac{d}{d-1}}}$.
\item[(c)] If $d\geq 3$ then
\beq{match}
\m(\G_1)=\min\set{|L_1|,\,|R_1|}=\min\set{n-\t_1,(1-(1+z^*)e^{-z^*})m+o(m)}.
\eeq
Here $L_1\subseteq L, R_1\subseteq R$ are the two sides of the bipartition of $\G_1$, after deleting any isolated vertices from the 
$R$-side.
\end{description}
\end{theorem}
\section{Structure of the paper}
We first prove Theorem \ref{th1}. This involves studying Phase 1 of the Karp-Sipser Algorithm. For this we first describe 
the distribution of the graph $G$. This is done in Section \ref{probsec}. The distribution of $\G$ is determined
by a few parameters and these evolve as a Markov chain. To study this chain, we introduce and solve a set of differential
equations. This is done in Section \ref{diffeqs}. We show that the chains trajectory and the solution to the equations are close.
By analysing the equations we can tell when Phase 1 is sufficient to solve the problem. This is done in Section \ref{ap1}.
If Phase 1 is not sufficient then the graph $\G_1$ that remains has degree $d$ on the $L$-side and minmum degree at least two
on the $R$-side. We show that \whp\ $\G_1$ has a matching of size equal to the minimum set size of the partition. 
\cite{DMP} and \cite{DM} and \cite{FP}. 
\section{Probability Model for Phase 1}\label{probsec}
We will represent $G$ and more generally $\G$ by a random sequence $\bx\in \Omega_{L,R}=(R^d\cup\set{\star}^d)^L$. 
A sequence $\bx\in \Omega_{L,R}$ is to be viewed as $n$ subsequences $\bx_1,\bx_2,\ldots,\bx_n$ where $\bx_j=(x_{j,1},x_{j,2},\ldots,
x_{j,d})\in R^d$
or $\bx_j=\s=(\star,\star,\ldots,\star)$. 
The $\star$'s represent edges that have been deleted by the Karp-Sipser algorithm.
For $\bx\in \Omega_{L,R}$ we 
define the bipartite (multi-)graph $\G_\bx$ as follows: Its vertex set consists of a bipartition
$L_\bx=\set{j\in L:\bx_j\neq \s}$ (the {\em left side}) and $R$ (the {\em right side}). The edges incident with $j\in L_\bx$ are 
$(j,x_{j,i}),\,i\in [d]$
i.e. we read the sequence $\bx$ from left to write and add edges to $\G_\bx$ in blocks of size $d$. Each block being
assigned to a unique vertex of $L_\bx$. 

We should be clear now that our probability space is $\Omega_{L,R}$ with uniform measure and not 
$\cG_{L,R}=\set{\G_\bx:\bx\in \Omega_{L,R}}$.

Given a graph
$\G_\bx$ we let $v_j=v_j(\bx)=|R_j(\bx)|$ where $R_j(\bx)$ is the set of vertices in $R$
that have degree $j\geq 0$. We let $v=v(\bx)=|R_\bx|$ where $R_\bx=\bigcup_{j\geq 2}R_j(\bx)$.

For the graph $G$ we choose $\bx(0)$ uniformly at random from $(R^d)^L$ and put $G=\G_{\bx(0)}$. Next
let $\G(0)=G$ and let $\G(t)=\G_{\bx(t)}$ be the graph $\G$ that we have after $t$ 
steps of Phase 1 of the Karp-Sipser algorithm. The sequence $\bx(t)$ is defined as follows: Observe first that
the vertex $\xi$ of degree one is always in $R$. Suppose that it is incident to the unique edge $(\eta,\xi),\eta\in L$. Then
we simply replace $\bx_\eta$ in $\bx(t-1)$ by $\s$ to obtain $\bx(t)$. We should thus think of the Karp-Sipser Algorithm as acting
on sequences $\bx$ and not on graphs. We write $\bx\to\by$ to mean that $\by$ can be obtained from $\bx$ by a single 
Phase 1 step of the Karp-Sipser Algorithm. 

Let $\vv(t) =
(w(t),v_1(t),v(t))$ where $w(t)=|L_{\bx(t)}|$ is the number of vertices on the left
side of the bipartition of $\G(t)$. Assuming that we have only run the Karp-Sipser algorithm up to the
end of Phase 1, we have $w(t) = n-t$. Also 
\beq{begin}
\vv(0)\sim (n,\a d e^{-\a d}m,(1-e^{-\a d}-\a d e^{-\a d})m)\ \ \whp. 
\eeq
We will omit the parameter
$t$ from $\vv(t)$ when it is clear from the context. 
Let $\Gv$ be
the set of all $\bx\in \Omega_{L,R}$ with parameters $\vv$.

\begin{lemma}\label{lem2}
Suppose $\bx(0)$ is a random member of $\cX_{\vv(0)}$. Then given $\vv(0),\ldots, \vv(t)$, $\bx(t)$
is a random member of $\cX_{\vv(t)}$ for all $t\ge 0$.
\end{lemma}
\begin{proof}
We prove this by induction on $t$. It is true for $t=0$ by assumption and so assume it is true for some $t\geq 0$.
Let $\vv(t)=(w,v_1,v)$ and now fix a triple $\vv'=(w'=w-1,v_1',v')$ as a possible value for $\vv(t+1)$. Fix $\by\in \cX_{(w',v_1',v')}$.
We first compute the
number of $\bx\in\cX_{\vv(t)}$ such that $\bx\to\by$. Let $b=v-v'$ be the number of vertices in $R_\bx\setminus R_\by$. Some of these
will be in $R_0(\by)$ and some will be in $R_1(\by)$. So we choose non-negative integers $b_0,b_1$ such that 
$b_0+b_1=b$. Next let $a=v_1-v_1'+b_1$ be the number of vertices in $R_1(\bx)\cap R_0(\by)$. 
We can choose a vertex $\eta$ so that $L_\bx\setminus L_\by=\set{\eta}$ in $t$\mstar\  
ways and now let us enumerate the ways of choosing $\bx_\eta=(\z_1,\z_2,\ldots,\z_d)$.
Our choices for $\z_i$ are (i) distinctly from $R_0(\by)$ (i.e. $\z_i$ is distinct from rest of the $\z_j$), (ii) 
non-distinctly from $R_0(\by)$ (i.e. $\z=\z_i$ is chosen more than once in the construction), (iii) from $R_1(\by)$ and 
(iv) from $R_\by$. We must exercise choice (i) exactly $a$ times, choice (ii) at least twice for each of $b_0$ distinct values,
choice (iii) at least once for each of $b_1$ distinct values and choice (iv) the remaining times.

\ignore{
Next let $c$
be the number of distinct neighbours of $\eta$ in \bx. We will use the notation
$|z|=z_1+\cdots+z_k$ for a vector $z=(z_1,\ldots,z_k)$. 
The number of choices for \bx\
is then 
\beq{indy}
\sum_{\substack{0\leq c\leq d-a-b\\0\leq b_1\leq b}}\ 
\sum_{\substack{\br\in [d]^{b_1},\bs\in [d]^{b_2},\bt\in [d]^c\\|r|+|s|=d-a-c}}(n-t)
\binom{m-v'-v_1'}{a}\binom{v_1'}{b}\binom{v'}{d-a-b}(v')^{d-a-b}\binom{d}{c}\binom{d-c}{r_1,\ldots,s_b}.
\eeq 
(i) $n-t$ choices for $\xi\notin L_\by$, (ii) $\binom{m-v'-v_1'}{a}$ choices for 
the set of vertices of degree one in \bx\ that become of degree zero in \by, (iii) $\binom{v_1'}{b}$ choices for 
the set of vertices of degree at least two in $\bx$ 
that become of degree one in $\by$, (iv) $\binom{v'}{d-a-b}$ choices for the set of vertices of degree at least two in \bx\
that are incident with $\xi$ and have degree at least two in \by, (v) $d!$ ways of filling in $\bx_\xi$. 
}
The number of choices for $\bx_\eta$ depends only on $\vv$ and $\vv'$,
i.e. for each $\by\in \cX_{\vv'}$ we have that 
$D(\vv,\vv') = \left|\set{\bx \in \cX_{\vv} : \bx \to \by }\right|$ 
is independent of \by, given $\vv$ and
$\vv'$. 

Similarly given $\bx$ there is a unique $i\in R_1(\bx)$, which when removed determines
$\by$. Thus $N(\vv) = \left|\set{\by : \bx \to \by }\right|$ is fixed given $\vv$.  Thus if $\bx(t)$ is a random member of $\cX_{\vv}$ then 
\begin{align*}
&\Pr(\bx(t+1) = \by | \vv(0),\ldots,\vv(t))\\ 
&= \sum_{\bx \in \cX_{\vv(t)}}\Pr(\bx(t) = \bx | \vv(0),\ldots,\vv(t))\cdot \Pr(\bx(t+1) = y | \vv(0),\ldots,\vv(t-1),\bx(t)=\bx)\\
&= \sum_{\bx\in\cX_{\vv(t)}} \Pr(\bx(t+1)=\by | \bx(t)=\bx) \cdot \left|\cX_{\vv(t)}\right|^{-1}\\
&= \sum_{\bx\in\cX_{\vv(t)}} \frac{D(\vv(t),\vv(t+1))}{N(\vv(t))} \cdot \left|\cX_{\vv(t)}\right|^{-1}
\end{align*}
which is independent of $\by$ given $\vv(t)$ and so $\by$ is a random member of $\cX_{\vv(t+1)}$.
\end{proof}

\begin{lemma}
The random sequence $\vv(t), t=0,1,2,\ldots$ is a Markov chain.
\end{lemma}
\begin{proof}
As in \cite{AFP},
\begin{eqnarray*}
\Pr(\vv(t+1)\mid \vv(0),\vv(1),\ldots,\vv(t))&=&\sum_{\bx'\in
\cX_{\vv(t+1)}}\Pr(\bx'\mid 
\vv(0),\vv(1),\ldots,\vv(t))\\
&=&\sum_{\bx'\in \cX_{\vv(t+1)}}\sum_{\bx\in \cX_{\vv(t)}}\Pr(\bx',\bx\mid 
\vv(0),\vv(1),\ldots,\vv(t))\\
&=&\sum_{\bx'\in \cX_{\vv(t+1)}}\sum_{\bx\in \cX_{\vv(t)}}
\Pr(\bx'\mid \vv(0),\vv(1),\ldots,\vv(t-1),\bx)\\
&&\times \Pr(\bx\mid 
\vv(0),\vv(1),\ldots,\vv(t))\\
&=&\sum_{\bx'\in \cX_{\vv(t+1)}}\sum_{\bx\in \cX_{\vv(t)}}\Pr(\bx'\mid
\bx)|\cX_{\vv(t)}|^{-1},
\end{eqnarray*}
which depends only on $\vv(t),\vv(t+1)$.
\end{proof}

\begin{lemma}\label{trunp}
Conditional on $\vv$ if $\bx$ is selected uniformly at random from
$\Gv$ then each vertex $i\in R_\bx$ has degree $Y_i$ where $Y_i=\Poi(z;\ge 2)$, a Poisson random variable
conditioned to take a value at least two, and $z$ satisifes
\begin{equation}\label{zeqn}
\frac{z(e^z-1)}{f(z)}=\frac{dw-v_1}{v}
\end{equation}
where $f(z) = e^z-z-1$.

The $Y_i$ are also conditioned to satisfy
$\sum_{i=1}^{v}Y_i = wd-v_1$. 
\end{lemma}
\begin{proof}
Suppose we first fix the edges incident with vertices of degree one in \bx. Then we randomly fill in the remaining
$dw-v_1$ non-$\star$ positions in \bx\ with values from some fixed $v$-subset $R_\bx$ of $R$, subject to each of these $v$ vertices
having degree at least two. The degrees $Y_i$ of these vertices will have the description described in the lemma (for a proof 
see Lemma 4 of \cite{AFP}). 
\end{proof}

From \cite{AFP} we can use the following lemma
\begin{lemma} {\bf \cite{AFP}}\\
{\bf (a)} Assume that $\log n = O((vz)^\half)$. For every $j\in R_\bx$ and $2\le k \le \log n$,
\beq{eq0}
\Pr(Y_j = k | \vv) = \frac{z^k}{k!f(z)}\brac{1+O\bfrac{k^2+1}{vz}}
\eeq
{\bf (b)} For all $k\ge 2, j\in R_\bx$
$$
\Pr(Y_j = k | \vv) = O\brac{(vz)^\half \frac{z^k}{k!f(z)}}
$$
\end{lemma}
\proofend
\section{Differential Equations}\label{diffeqs}
Let $\vv$ be the current parameter tuple and $\vv'$ be the tuple after one step of the Karp-Sipser algorithm. The following
lemma gives $E[\vv'-\vv|\vv]$ for each step of Phase 1 of the Karp-Sipser algorithm.
\begin{lemma}\label{change}
Assuming $\log v = O((vz)^\half)$ and $v_1>0$ we have
\begin{align*}
\E[v_1'-v_1 | \vv] &= -1 - \frac{d-1}{dw}v_1 + \frac{d-1}{dw}\frac{vz^2}{f} + O\bfrac{1}{vz} \\
\E[v'-v | \vv] &= -\frac{d-1}{dw}\frac{vz^2}{f} + O\bfrac{1}{vz}
\end{align*}
\end{lemma}
\begin{proof}
  First note that $v_1>0$, one vertex $\xi\in R$ with $\deg(\xi) = 1$ will
  be picked and $\xi$ and its neighbor $\eta\in L$ will be removed from
  $G_t$. This implies that $w$ decreases by $1$ and the number of
  edges removed is $d$, i.e. all edges incident to $\eta$. Let $\d$ be
  the number of multiple edges incident to $\eta$. Then we have
\begin{align}
\E[\d | \vv] &\le d\cdot \Pr(\eta \text{ is incident to parallel edges}) \le d \binom{d}{2} \frac{1}{(wd)_2}\sum_{j}\E[\deg(j)_2]\label{eq1}\\
&\le \frac{d^3}{2w^2}v\E[Y(Y-1) | \vv] = O\bfrac{1}{w}\nonumber
\end{align}
where $Y$ has distribution \eqref{eq0}.

{\bf Explanation:}
The $dw$ choices of neighbours for the remaining vertices in $L$ form a list with $v_1$ unique names and $wd-v_1$ non-unique
names and where the number of times a vertex appears among the $wd-v_1$ has distribution \eqref{eq0}. Also, if we construct this 
list vertex by vertex, it will appear in a random order. So the probability that $j$ appears in two of the choices for 
$\eta$ is bounded by $\frac{\E[\deg(j)_2]}{(wd)_2}$ and this justifies \eqref{eq1}.

The change in $v_1$ comes from $\xi$ being removed, minus the number of 
other degree one vertices adjacent to $\eta$ and plus the number of vertices
adjacent to $\eta$ of degree exactly two. Any change from vertices of degree three
or more is absorbed by the $O\bfrac{1}{w}$ term for multiple edges.

The expected change is then
\begin{align}
\E[v_1'-v_1 | \vv] &= -1 -\frac{d-1}{dw-1}(v_1-1) + \frac{2(d-1)}{dw-1}\E[v_2 | \vv] + O\bfrac{1}{w}\nonumber\\
&= -1 -\frac{d-1}{dw}v_1 + \frac{d-1}{dw}\frac{z^2}{f}v + O\bfrac{1}{vz} \label{v1cc}
\end{align}

Similarly for $v$, the change is only due to vertices adjacent to $y$ of degree exactly two, modulo multiple edges. Thus
\begin{align*}
\E[v'-v | \vv] &= -\frac{d-1}{dw}\frac{z^2}{f}v + O\bfrac{1}{vz}
\end{align*}
\end{proof}
\\\
Lemma \ref{change} suggests that we consider the following pair of differential equations
\begin{align}
\frac{dy_1}{dt} &= -1 -\frac{d-1}{dw}y_1 + \frac{d-1}{dw}\frac{y\z^2}{f(\z)}\label{dv1}\\
\frac{dy}{dt} &=  -\frac{d-1}{dw}\frac{y\z^2}{f(\z)}\label{dv}
\end{align}
where $w = n-t$ and $\z$ satisfies 
\beq{13}
\frac{\z(e^\z-1)}{f(\z)}=\frac{dw-y_1}{y}.
\eeq
The boundary conditions are (see \eqref{begin})
\beq{bound}
\z(0)=\a d,\;y_1(0)= m \a de^{-\a d},\;y(0)= m(1-(1+\a d)e^{-\a d}).
\eeq
The $y_1,y,\z$ are of course the deterministic counterparts of $v_1,v,z$ respectively.

\begin{lemma}\label{lemde}
The solution to \eqref{dv1}, \eqref{dv} and \eqref{bound} is
\begin{eqnarray}
w&=&\bfrac{\z}{\a d}^{\frac{d}{d-1}}n.\label{zw}\\
t&=&n\brac{1-\bfrac{\z}{\a d}^{\frac{d}{d-1}}}\label{zw1}\\
y&=& e^{-\z}f(\z)m\label{v-form}\\
y_1&=&m\z\brac{\bfrac{\z}{\a d}^{\frac{1}{d-1}} + e^{-\z}-1}\label{sad}
\end{eqnarray}
\end{lemma}
\begin{proof}
We take the derivative of \eqref{zeqn} with respect to $t$.
The RHS becomes, using \eqref{dv1} and \eqref{dv}
\begin{align}
\frac{d}{dt}\bfrac{dw-y_1}{y} &= \frac{1}{y}\brac{-d-\brac{-1-\frac{d-1}{d}\frac{y_1}{w} + \frac{d-1}{d}\frac{y}{w}\frac{\z^2}{f(\z)}} - 
\frac{dw-y_1}{y}\frac{1}{y}\brac{-\frac{d-1}{d}\frac{y}{w}\frac{\z^2}{f(\z)}}}\nonumber\\
&= -\frac{d-1}{dw}\brac{\frac{dw-y_1}{y} + \frac{\z^2}{f(\z)}-\frac{dw-y_1}{y}\frac{\z^2}{f(\z)}}\nonumber\\
&= -\frac{d-1}{dw}\bfrac{\z(e^\z-1)f(\z)+ \z^2f(\z)-\z^3(e^\z-1)}{f(\z)^2}\nonumber\\
&= -\frac{d-1}{dw}\bfrac{\z(e^\z-1)^2 - \z^3e^\z}{f(\z)^2}\label{a}
\end{align} 

On the other hand, on differentiating the LHS of \eqref{zeqn} (with $z$ replaced by $\z$) we get
\begin{align}
\frac{d}{dt}\bfrac{\z(e^\z-1)}{f(\z)} &= \frac{(\z e^\z+e^\z-1)f(\z) - \z(e^\z-1)^2}{f(\z)^2}\frac{d\z}{dt}\nonumber\\
&= \frac{(e^\z-1)^2-\z^2e^\z}{f(\z)^2}\frac{d\z}{dt}.\label{b}
\end{align}

Comparing \eqref{a} and \eqref{b} we see that
\begin{equation}\label{zweqn}
\frac{1}{\z}\frac{d\z}{dt} = -\frac{d-1}{dw}
\end{equation}

Integrating yields
\begin{equation*}
\frac{\z^d}{w^{d-1}} = \textrm{constant}
\end{equation*}

Plugging in $\z(0) = \a d$ we see that 
$$\frac{\z^d}{w^{d-1}} = \frac{(\a d)^d}{n^{d-1}}.$$
This verifies \eqref{zw} and \eqref{zw1}.

Going back to \eqref{dv} and \eqref{zweqn} we have 
$$\frac{dy}{d\z}\frac{d\z}{dt} = \frac{dy}{dt} = -\frac{(d-1)y}{dw}\frac{\z^2}{f(\z)}\\
= \frac{1}{\z}\frac{d\z}{dt}\frac{y\z^2}{f(\z)}.$$
So 
$$\frac{1}{y}\frac{dy}{d\z} = \frac{\z}{f(\z)}$$ 
and integrating yields
\begin{equation*}
\ln y = -\z + \ln f(\z) + C
\end{equation*}
Taking $y(0) = m(1-(1+\a d)e^{-\a d})$ gives $C = \ln m$
and we see that \eqref{v-form} holds.

{\large
{\bf
\ignore{
Using \eqref{v-form} in \eqref{dv1} we get
\begin{align*}
\frac{dy_1}{dz}\frac{dz}{dt} &= \frac{dy_1}{dt} \\
&= -1 -\frac{d-1}{dw}y_1 +\frac{d-1}{dw}\frac{vz^2}{f(z)}\\
&= -1 + \frac{1}{z}\frac{dz}{dt}y_1 - \frac{1}{z}\frac{dz}{dt}\frac{vz^2}{f(z)}\\
&= -1 + \frac{dz}{dt}\frac{y_1}{z}-\frac{dz}{dt}ze^{-z}m
\end{align*}
Dividing through by $\frac{dz}{dt}$ we get
\begin{equation*}
\frac{dy_1}{dz} = \frac{dw}{(d-1)z} + \frac{y_1}{z} - ze^{-z}m
\end{equation*}
and dividing by $z$ and collecting $y_1$ on LHS gives
\begin{align*}
\frac{d}{dz}\brac{\frac{y_1}{z}} &= \frac{dw}{(d-1)z^2} - e^{-z}m\\
&= \frac{dn}{\bfrac{nd}{m}^{\frac{d}{d-1}}}\frac{z^{\frac{1}{d-1}-1}}{d-1} - e^{-z}m
\end{align*}
Integrating with respect to $z$ and pluggin in $y_1(0) = z(0)e^{-z(0)}m$ gives
\begin{equation}\label{v1-form}
\frac{y_1}{z} = \frac{dn}{\bfrac{nd}{m}^{\frac{1}{d-1}}} + e^{-z}m
\end{equation}
Do we need to derive \eqref{v1-form}? It is not used for the next derivation.
}
}
}

We now solve for $y_1$ in terms of $\z$ as a function of $d$. It follows from \eqref{13} and 
\eqref{v-form} that
\begin{align*}
y_1 &= dw+\z e^{-\z}m-m\z \\
&= \bfrac{\z}{\a d}^{\frac{d}{d-1}}nd + \z e^{-\z}m -m\z\\
&= m\z\brac{\bfrac{\z}{\a d}^{\frac{1}{d-1}} + e^{-\z}-1}
\end{align*}
\end{proof}

At this point we wish to show that \whp\ the sequence $\vv(t),t\geq 0$ closely follows the trajectory
$\yy(t)=(w,y_1,y),t\geq 0$
described in Lemma \ref{lemde}. One possibility is to use Theorem 5.1 of Wormald \cite{Wormald}, but there
is a problem with an ``unbounded'' Lipschitz coefficient. 
One can allow for this in \cite{Wormald}, but it is unsatisfactory to ask the reader to 
check this. We have decided to use an approach suggested in Bohman \cite{r3t}. 

Next let $K$ be a large positive constant and let $\g\ll 1/K$.
Then let
$$g(x)=(1-x)^{-K}+K(1-x)^{-1}$$
and
$$Err(t)=n^{2/3}g(t/n).$$
Then define the event
$$\cE(t)=\set{v(\t)z(\t)\geq n^{1/2}\text{ and }\z(\t)\geq \z(t_1)+n^{-\g}\text{ and }\card{\vv(\t)-\yy(\t)}_\infty
\leq 2Err(\t)\ for\ \t\leq t}$$
where 
\beq{t1}
t_1=\min\set{t>0:\;y_1(t)=0}.
\eeq
Now define four sequences of random variables: 
\begin{eqnarray*}
X_1^{\pm}(t)&=&\begin{cases}v_1(t)-y_1(t)\pm Err(t)&\cE(t-1)\ holds\\
                X_1^{\pm}(t-1)&otherwise
               \end{cases}\\
X^{\pm}(t)&=&\begin{cases}v(t)-y(t)\pm Err(t)&\cE(t-1)\ holds\\
                X^{\pm}(t-1)&otherwise
               \end{cases}
\end{eqnarray*}
Because $(1-x)^{-L}$ is convex we have 
$$\frac{1}{(1-(x+h))^L}\geq \frac{1}{(1-x)^L}+\frac{hL}{(1-x)^{L+1}}$$
for $L>0$ and $0<x<x+h<1$. So,
\beq{gt}
g((t+1)/n)-g(t/n)\geq \frac{K}{n-t}\,g(t/n).
\eeq

Suppose that $\cE(t)$ holds. We write
\begin{eqnarray}
\card{\frac{dw-y_1}{y}-\frac{dw-v_1}{v}}&=&\card{\frac{(dw-y_1)(v-y)}{yv}+\frac{v_1-y_1}{v}}\nonumber\\
&=&O\bfrac{Err(t)}{v}.\label{rat}
\end{eqnarray}
(For this we need $(dw-y_1)/y=O(1)$. But this follows from \eqref{13} and the fact that $\z$ is decreasing -- see \eqref{zweqn}).

Putting $F(x)=\frac{x(e^x-1)}{f(x)}$ 
we have (see \eqref{b}) $F'(x)=\frac{(e^x-1)^2-x^2e^x}{f(x)^2}$ and since
\beq{diff}
(e^x-1)^2-x^2e^x=\sum_{k=4}^\infty(2^k-2-k(k-1))\frac{x^k}{k!}
\eeq
we see that $F'(x)=\Omega(1)$ in any bounded interval $[0,L]$. 

Hence from \eqref{rat} we have
$$O\bfrac{Err(t)}{v}=F(\z)-F(z)=\Omega(|\z-z|)$$
or
\beq{zzz}
|\z-z|=O\bfrac{Err(t)}{v}.
\eeq
Using \eqref{zzz} we obtain
\begin{eqnarray}
\card{\frac{vz^2}{f(z)}-\frac{y\z^2}{f(\z)}}&\leq&\frac{|v-y|z^2}{f(z)}+y\card{\frac{z^2}{f(z)}-\frac{\z^2}{f(\z)}}\nonumber\\
&\leq &K_1Err(t)\label{plm}
\end{eqnarray}
for some $K_1=K_1(\a,d)>0$.

For the second term we use 
$$\bfrac{x^2}{f(x)}'=\frac{2xf(x)-x^2(e^x-1)}{f(x)^2}=-\frac{1}{f(x)^2}\sum_{k=4}^\infty\frac{k-3}{(k-1)!}x^k.$$ 
This implies that $\bfrac{x^2}{f(x)}'=O(1)$ for $x\geq 0$.

Now with $f=f(\z)$,
\begin{align*}
y_1''(t)&=
\frac{d-1}{dw}\brac{1-\frac{y_1}{w}+\frac{y\z^2}{wf}+\frac{d-1}{dw}\brac{y_1-\frac{y\z^2}{f}-\frac{y\z^4}{f^2}-
y\z^2\brac{\frac{2}{f}
-\frac{e^\z-1}{f^2}}}}\\
&=O(\z^{-O(1)}n^{-1}).\\
y''(t)&=
\frac{d-1}{dw}\brac{\frac{y\z^2}{wf}+\frac{d-1}{dw}\brac{\frac{y\z^4}{f^2}+y\z\brac{\frac{2\z}{f}
-\frac{\z(e^\z-1)}{f^2}}}}\\
&=O(\z^{-O(1)}n^{-1}).
\end{align*}
If $\cE(t)$ holds then, where $\r_t=\frac{d-1}{dw}=\frac{d-1}{d(n-t)}$,
\begin{align*}
&\E(X_1^+(t+1)-X_1^+(t)\mid \vv(t))=\\
&\E(v_1(t+1)-v_1(t)\mid \vv(t))-(y_1(t+1)-y_1(t))+n^{2/3}(g((t+1)/n)-g(t/n))\geq\\
&- \r_{t}(v_1(t)-y_1(t)) + \r_{t}\brac{\frac{v(t)z(t)^2}{f(z(t))}-\frac{y(t)\z(t)^2}{f(\z(t))}} 
+ O\bfrac{1}{v(t)z(t)} - y_1''(t+\th)+\frac{Kn^{2/3}}{n-t}\,g(t/n)\\
&\geq \frac{n^{2/3}g(t/n)}{n-t}\brac{-K_1-2+K}+ O\bfrac{1}{v(t)z(t)} - y_1''(t+\th)\\
&\geq 0.
\end{align*}

This shows that $X_1^+(t),t\geq 0$ is a sub-martingale. Also,
\begin{multline*}
|X_1^+(t+1)-X_1^+(t)|\leq \\
|v_1(t+1)-v_1(t)|+\sup_{0\leq\th\leq 1}|y_1'(t+\th)|+n^{2/3}(g((t+1)/n)-g(t/n))=O(1 ).
\end{multline*}

It follows from the Azuma-Hoeffding inequality that we can write
$$\Pr(\exists 1\leq t\leq t_1:X_1^+(t)\leq X_1^+(0)-n^{3/5})\leq e^{-\Omega(n^{1/5})}.$$
By almost identical arguments we have
\begin{align*}
&\Pr(\exists 1\leq t\leq t_1:X_1^-(t)\geq X_1^-(0)+n^{3/5})\leq e^{-\Omega(n^{1/5})}.\\
&\Pr(\exists 1\leq t\leq t_1:X^+(t)\leq X^+(0)-n^{3/5})\leq e^{-\Omega(n^{1/5})}.\\
&\Pr(\exists 1\leq t\leq t_1:X^-(t)\geq X^-(0)+n^{3/5})\leq e^{-\Omega(n^{1/5})}.\\
\end{align*}
It follows that \whp, when $\cE(t)$ holds, we have
\begin{align}
&|v_1(t)-y_1(t)|\leq Err(t)+n^{3/5}+|v_1(0)-y_1(0)|<2Err(t).\label{tite1}\\
&|v(t)-y(t)|\leq Err(t)+n^{3/5}+|v(0)-y(0)|<2Err(t).\label{tite2}
\end{align}
Now by construction, $\cE(t)$ will fail at some time $t_2\leq t_1$. It follows from \eqref{tite1}, \eqref{tite2} that \whp\ 
it will fail either because (i) $v(t_2)z(t_2)<n^{1/2}$ or $\z(t_2)<\z(t_1)+n^{-\g}$. We claim the latter.
Observe that if $\z(\t)\geq \z(t_1)+n^{-\g}$ then \eqref{zw}--\eqref{sad} imply $w,y_1=\Omega(n^{1-d\g/(d-1)})$ and
$y=\Omega(n^{1-2\g})$. Together with \eqref{tite1}, \eqref{tite2}, this implies that $v(t_2)z(t_2)=\Omega(n^{1-3\g})
\geq n^{1/2}$.

In summary then, \whp\, the process satisfies
\beq{u1}
|\vv(t)-\yy(t)|_\infty<2Err(t_2)=O\bfrac{n^{2/3}}{(1-t_2/n)^K}=O(n^{2/3+Kd\g/(d-1)}) \qquad\qquad for\ 1\leq t\leq t_2
\eeq
and
\beq{u2}
v_1(t_2)=O(n^{1-\g})\text{ where }t_2=t_1+O(n^{1-\g}).
\eeq
We use \eqref{zw} for \eqref{u1} and \eqref{sad}, \eqref{zzz} for \eqref{u2}.

\ignore{
\section{Growth of $v_1$}
When $v_1$ is small, we see that \eqref{dv1} becomes
$$\frac{dv_1}{dt} = -1  + \frac{d-1}{dw}\frac{vz^2}{f(z)}=-1+\frac{(d-1)z}{e^z-1}$$
after substituting for $dw$ using \eqref{actual}.

This is positive for small $z$. This is because our differential equations are only good for Phase 1.
In Phase 2, vertices on the left lose edges. I think we observed this already. To handle Karp-Sipser in
Phase 2 we need to keep track of the number of vertices on the left with degree $j$ for $2\leq j\leq d$. So we
need $d-1$ variables $v,v_{j,left},j=2,\ldots,d$ and $d-1+2$ equations where the extra two equations are needed to keep track of the growth of
$v_1$ and $v_{1,left}$.
}
\section{Analysis of Phase 1}\label{ap1}
We will first argue that \whp\ Phase 1 is sufficient to find a matching from $L$ to $R$ when there is no solution $0<\z\leq \a d$ to 
\beq{p1e}
\bfrac{\z}{\a d}^{\frac{1}{d-1}} + e^{-\z}-1=0.
\eeq
It follows from \eqref{sad}, \eqref{t1} and \eqref{u2} that in this case Phase 1 ends with there being at most
$O(n^{1-\g})$ vertices of $L$ left unmatched, \whp. Furthermore at time $t_2$ we will have
$$w\sim \bfrac{\z}{\a d}^{\frac{d}{d-1}}n,\,v_1\sim dw\text{ and }v=O(\z^{\frac{d-2}{d-1}}v_1)$$
where $\z=\z(t_2)=O(n^{-\g})$.
\begin{lemma}\label{lemP1}
Suppose that $t_1=0$. Then \whp\ at time $t_2$, $\G$ is a forest.
\end{lemma}
\begin{proof}
Let $R_{\geq 2}$ denote the set of vertices of degree at least two in
the $R$-side of $\G$. 
Let $P(d_1,\ldots,d_k)$ denote $\Pr(X_1=d_1,\ldots,X_k=d_k)$ where $X_1,\ldots,X_v$ are truncated Poisson conditioned
only to sum to $dw-v_1$. For large $k$ we use the bound 
\beq{b1}
\Pr(X_1=d_1,\ldots,X_k=d_k)\leq O(n^{1/2})\prod_{i=1}^k\frac{z^{d_i}}{d_i!f(z)}.
\eeq
For $k=O(1)$ and $d_1,\ldots,d_k=O(\log n)$ we write
\beq{b2}
\Pr(X_1=d_1,\ldots,X_k=d_k)=\prod_{i=1}^k\Pr(X_i=d_i\mid X_j=d_j,j<i)=(1+o(1))\prod_{i=1}^k\frac{z^{d_i}}{d_i!f(z)}.
\eeq
It is equation \eqref{eq0} that alows us to write the final equality in \eqref{b2}. The extra conditioning
$X_j=d_j,j<i$ only changes the required sum.

Thus let $B_k$ denote $O(n^{1/2})$ for $k\geq \frac{2(d-1)}{(d-2)\g}$ and $1+o(1)$ otherwise.
The expected number of cycles can be bounded by $o(1)=(\Pr(\exists\text{ vertex of degree }\geq\log n))$ plus
\begin{align}
&\sum_{k\geq 2}\sum_{\substack{S\subseteq R_{\geq  2}\\|S|=k}}\sum_{2\leq d_1,\ldots,d_k\leq \log n}P(d_1,\ldots,d_k)
\times\binom{w}{k}(d(d-1))^k(k!)^2\prod_{i=1}^k\frac{d_i(d_i-1)}{(dw-2i+2)(dw-2i+1)}\leq \label{cycles}\\
&\sum_{k\geq 2}B_k\binom{v}{k}\sum_{2\leq d_1,\ldots,d_k\leq \log n}\prod_{i=1}^k\frac{z^{d_i}}{d_i!f(z)}
\times\binom{w}{k}\binom{d}{2}^k(k!)^2\prod_{i=1}^k\frac{d_i(d_i-1)}{(dw-2i+2)(dw-2i+1))}\leq \nonumber\\
&\sum_{k\geq 2}B_k\bfrac{vwz^2d^2}{f(z)(dw-2k)^2}^k\sum_{2\leq d_1,\ldots,d_k}\prod_{i=1}^k\frac{z^{d_i-2}}{(d_i-2)!}=
\nonumber\\
&\sum_{k\geq 2}B_k\bfrac{vwz^2e^zd^2}{f(z)(dw-2k)^2}^k=\nonumber\\
&\sum_{k\geq 2}B_kO(\z^{\frac{d-2}{d-1}})^k=\nonumber\\
&o(1).\nonumber
\end{align}
\end{proof}

{\bf Explanation of \eqref{cycles}:} We condition on the degree sequence. 
Having fixed the degree sequence,
we swap to the configuration model \cite{BoCo}. Having chosen $S\subseteq R$ and $k$ vertices $W$ in $L$ and their
degrees, we can work within tnis model. We then choose a cycle through these vertices in $(k!)^2$ ways.
We then choose the configuration points associated with our $k$-cycle in 
$(d(d-1))^k\prod_{i=1}^kd_i(d_i-1)$ ways. We then multiply by the probability $\prod_{i=1}^k\frac{1}{(dw-2i+2)(dw-2i+1)}$
of choosing the pairings associated with the edges of the cycle.
\begin{corollary}\label{cor1}
Suppose that $t_1=0$. Then \whp\ at time $t_2$, $\G$ contains a matching from $L_\G$ into $R_\G$.
Furthemore, such a matching will be constructed in Phase 1.
\end{corollary}
\begin{proof}
We can assume from Lemma \ref{lemP1} that $\G$ is a forest. Each vertex of $L_\G$ has degree
$d$ and so Hall's theorem will show that the required matching exists. (Any Hall witness would induce a cycle).
Finally note that Phase 1 of the Karp-Sipser algorithm is exact on a forest.
\end{proof}
\subsection{Threshold for Phase 1 to be sufficient}
Put $A=(\a d)^{-1/(d-1)}$ and $B=1/(d-1)$ so that \eqref{p1e} can be written as
\beq{za}
A\z^B + e^{-\z}-1=0.
\eeq
Assume $B$ is fixed. We find a threshold for $A$ in terms of $B$ for there to be no positive solution to \eqref{za}. 
We find the place $\z^*$ where the curve
$y=1-e^{-\z}$ touches the curve $y=A\z^B$ i.e. where
\begin{eqnarray*}
A\z^B&=&1-e^{-\z}\\
AB\z^{B-1}&=&e^{-\z}
\end{eqnarray*}
In which case
\beq{mm0}
\frac{\z^*}{B}=e^{\z^*}-1\text{ and }A^*=\frac{1-e^{-\z^*}}{(\z^*)^B}
\eeq
or
\beq{zaaz}
\z^*=\frac{e^{\z^*}-1}{d-1}\text{ and }\a^*=\frac{\z^*}{d(1-e^{-\z^*})^{d-1}}.
\eeq
In general, keeping $B<1$ fixed let 
$$f_A(\z)=A\z^B+e^{-\z}-1\text{ and }L_A=A^{-1/B}.$$
We must show that if $A\geq A^*$ then the only solution to $f_A(\z)=0,\,0\leq L_A$ is $\z=0$.

Observe that $f_A(L_A)=e^{-L_A}>0$ and $f_A(0)=0$. Also, if $A<1$ then $1<L_A$ and if $A<1-e^{-1}$ then $f_A(1)<0$
and there must be a positive solution to $f_A(\z)=0$.

Observe that $A'>A$ implies 
(i) $L_{A'}<L_A$ and that (ii) $f_{A'}(\z)>f_A(\z)$ for $\z\neq 0$. So if $f_A(\z)$ has no positive solution
then neither has $f_{A'}$. We argue that $f'_A(\z)=0$ has at most 2 solutions, which implies that $f_A(\z)=0$ has at most two 
positive solutions. As we increase $A$ to $A^*$ these solutions must converge, by (ii).

Now 
$$f_A'(\z)=0\text{ iff }\frac{e^\z}{\z^{1-B}}=\frac{1}{AB}.$$
But the function $g(\z)=e^\z/\z^\xi$ is convex for any $\xi>0$. Indeed
$$g''(\z)=\frac{e^\z((\z-\xi)^2+\xi)}{\z^{\xi+2}}>0$$
and so $g(\z)=a$ has at most two solutions for any $a>0$.
\subsection{Finishing the proof of Theorem \ref{th1}}
We now have to relate the above results to the actual process. We know from our analysis of the differential equations that
for some $A>0$,
$$v_1(t_2=t_1+An^{-\g})=O(n^{1-\g}).$$
When $\z_1=\z(t_1)=0$, Lemma \ref{lemP1} and Corollary \ref{cor1} imply Theorem \ref{th1}.

So assume that $\z_1>0$. Thus $\a>\a^*$ and $A<A^*$.
We argue that if $z_A$ is the solution to \eqref{za} then $z_A$ decreases monotonically with $A$.
Indeed, if $A'>A$ then $f_{A'}(\z)>f_A(\z)$ for $\z_A\leq \z\leq L_{A'}<L_A$.
Now $\frac{(d-1)z}{e^z-1}$ is strictly monotone decreasing with $z$ and so 
\beq{mm1}
\frac{(d-1)z_A}{e^{z_A}-1}<\frac{(d-1)z_{A^*}}{e^{z_{A^*}}-1}=1.
\eeq
The second equation in \eqref{mm1} is the first equation in \eqref{mm0}.

At time $t_2$ we will have $v_1=O(n^{1-\g}))$ and $v=\Omega(n)$. For the next $o(n)$ steps
we have from \eqref{v1cc} that
\beq{mm2}
\E(v_1'-v_1\mid \vv)=-1+(1+o(1))\frac{d-1}{dw}\frac{z^2}{f}v=-1+(1+o(1))\frac{(d-1)z_A}{e^{z_A}-1}\leq -\e
\eeq
for some small positive $\e$. In which case, \whp, $v_1$ will become zero in $O(n^{1-\g}\log n)$ steps.
Indeed \eqref{mm2} implies that the sequence
$$X_k=\begin{cases}v_1(t_2+k)+\e k&if\ v_1(t_2+k)>0\\X_{k-1}&otherwise\end{cases}$$
is a supermartingale that cannot change by
more than $d$ in any step. The Azuma-Hoeffding inequality implies that for $T=2X_0/\e$ we have
$$\Pr(v_1(t_2+T)>0)\leq \Pr(X_T-X_0\geq \e T-X_0)\leq \exp\set{-\frac{(\e T-X_0)^2}{2Td^2}}=o(1).$$
 
I.e. \whp\ $v_1(t_2+T)=0$ and $v(t_2+T)=v(t_2)-o(n)=\Omega(n)$.
\section{Proof of Theorem \ref{th2}}\label{pot2}
Let us summarize what we have to prove.
We have a random bipartite graph $\G_1$ with partition $L_1,R_1$ and $|L_1|=n_1=\a_1 m_1,|R_1|=m_1$.
Each vertex in $L_1$ has degree $d$ and each vertex in $R$ has degree at least 2.
At this point it is convenient to drop the suffix 1. So from now on, $m,n,\a,\G$ etc. refer to the graph left at the end of Phase 1.

The degrees of $\G$ satisfy, $d_{L}(a)=d$ for $a\in L$. The degrees of vertices in $R$ are distributed as the 
box occupancies $X_1,X_2,\ldots,X_n$ in the following experiment. We throw $dn$ balls randomly into $n$ boxes and condition that 
each box gets
at least two balls.  In these circumstance the $X_j$'s are independent truncated Poisson, subject to the 
condition that $X_1+X_2+\cdots+X_n=dn$, see Lemma \ref{trunp} with $v_1=0$. 
Thus for any $S=\set{b_1,b_2,\ldots,b_s}\subseteq R$ and any set of positive integers $k_i\geq 2,i\in S$ we have 
$$\Pr(d_{R_1}(b_i)=k_i,i\in S)\leq O(n^{1/2})\prod_{i\in S}\frac{z^{k_i}}{k_i!f(z)}$$ 
for 
$k\geq 2$ where $z$ satisfies 
$$\frac{z(e^{z}-1)}{f(z)}=\frac{nd}{m}.$$
The $O(n^{1/2})$ term accounts for the conditioning $\sum_{b\in R}d_{R}(b)=dn$
We will prove
\begin{theorem}\label{min}
Let $\G$ be a bipartite graph chosen uniformly from the sets of graphs with bipartition $L,R$, $|L|=n,|R|=m$
such that each vertex of $L$ has degree $d\geq 4$ and each vertex of $R$ has degree at least two. Then \whp
$$\m(\G)=\min\set{m,n}.$$
\end{theorem}
\subsection{Useful Lemmas}
Define the function $\z(\g),\g>0$ to be the unique solution to 
$$\frac{u(e^u-1)}{f(u)}=\g.$$
Let $g$ be defined by
$$g(x) = (e^{\z(x)}-1)^xf(\z(x))^{1-x}.$$
Observe that
\beq{gz}
\frac{f(\z(x))}{\z(x)^x}=\frac{g(x)}{x^x}.
\eeq
\begin{lemma}\label{use1}
The function $g(x)$ is log-concave as a function of $x$
\end{lemma}
\begin{proof}
We will write $\z$ for $\z(x)$ and $f$ for $f(\z)$ throughout this proof. Now $\frac{\z(e^\z-1)}{f} = x$ from
which we get 
\beq{dzx}
\frac{d\z}{dx} = \frac{f^2}{(e^\z-1)^2-\z^2e^\z}
\eeq
and note that $\frac{d\z}{dx}>0$ for $\z>0$. Taking the derivative of $\log(g(x))$ we get
\begin{align*}
\frac{d}{dx}\log(g(x)) &= \frac{d}{dx}\brac{x\log(e^\z-1)+(1-x)\log(e^\z-\z-1)}\\
&= \log\bfrac{e^\z-1}{f} + \frac{d\z}{dx}\brac{x \frac{e^\z}{e^\z-1}+(1-x)\frac{e^\z-1}{f}}
\end{align*}
Now $x = \frac{\z(e^\z-1)}{f}$ so 
\begin{align*}
x \frac{e^\z}{e^\z-1}+(1-x)\frac{e^\z-1}{f}&= \frac{\z e^\z}{f}+\frac{f-\z(e^\z-1)}{f}\frac{e^\z-1}{f}\\
&= \frac{\z e^\z(e^\z-\z-1)+(e^\z-\z-1 - \z e^\z+\z)(e^\z-1)}{f^2}\\
&= \frac{(e^\z-1)^2-\z^2e^\z}{f^2}\\
&= \frac{dx}{d\z}
\end{align*}
Thus we have
\beq{gdash}
\frac{d}{dx}\log(g(x)) = \log\bfrac{e^\z-1}{f}+1
\eeq

Taking the second derivative we get
\begin{align*}
\frac{d^2}{dx^2}\log(g(x)) &= \frac{d}{dx}\brac{\log\bfrac{e^\z-1}{f}+1}\\
&= \frac{f}{e^\z-1}\frac{e^z(e^\z-\z-1)-(e^\z-1)^2}{f^2}\frac{d\z}{dx}\\
&= \frac{1}{(e^\z-1)f}\frac{d\z}{dx}\brac{-(\z-1)e^\z-1}
\end{align*}
and since $-(\z-1)e^\z-1$ is strictly negative for $\z>0$ we get that $g(x)$ is log-concave
\end{proof}
\begin{lemma}\label{use2}
$\z(x)$ is concave as a function of $x$.
\end{lemma}
\begin{proof}
We begin with \eqref{dzx}. We note from \eqref{diff} that the denominator
$$(e^\z-1)^2-\z^2e^\z\geq 0.$$
Then we have
$$\frac{d^2\z}{dx^2}=\frac{2(1+\z)+e^\z(-6-e^{2\z}(2+\z(\z-4))+\z^2(5+\z(\z+2))+e^\z(6-2\z(2\z+3)))}{((e^\z-1)^2-\z^2e^\z)^2}
\frac{d\z}{dx}.$$
Now let
$$\f(u)=\sum_{n=0}^\infty \f_nu^n=2(1+u)+\psi(u)$$
where
$$\psi(u)=\sum_{n=0}^\infty\psi_nu^n=e^u(-6-e^{2u}(2+u(u-4))+u^2(5+u(u+2))+e^u(6-2u(2u+3))).$$
We check that $\psi_0=-2$ and $\psi_1=0$ which implies that $\f_0=\f_1=0$. One can finish the argument by checking that
$$\psi_n=-\frac{3^{n-2}(n^2-13n+18)+2^n(n^2+2n-6)-(n^4-4n^3+10n^2-7n-6)}{n!}\leq 0$$ 
for $n\geq 2$. This is simply a matter of checking for small values until the $3^n$ term dominates.
\end{proof}
Next let 
$$
H(u) = \log f(u)-u-2\log u = \log\bfrac{e^u-u-1}{u^2e^u}
$$
\begin{lemma}\label{use3}
$H(u)$ is convex as a function of $u$.
\end{lemma}
\begin{proof}
\begin{align*}
\frac{d^2}{du^2}H(u) &= \frac{d}{du}\brac{\frac{e^u-1}{f(u)}-1-\frac{2}{u}}\\
&= \frac{e^u(e^u-1-u) - (e^u-1)^2}{f^2(u)} + \frac{2}{u^2}\\
&= \frac{e^u-1 - ue^u}{f^2(u)} + \frac{2}{u^2}\\
&= \frac{u^2(e^u-1-ue^u)+2(e^u-1-u)^2}{u^2f^2(u)}\\
&= \frac{2e^{2u}+u^2e^u+u^2+4u+2-u^3e^u-4ue^u-4e^u}{u^2f^2(u)}
\end{align*}
Let
$$\f(u)=2e^{2u}+u^2e^u+u^2+4u+2-u^3e^u-4ue^u-4e^u=\sum_{n=0}^\infty \f_nu^n.$$
Direct computation gives $\f_0=\f_1=\f_2=0$ and for $n\geq 3$ 
$$\f_n=\frac{1}{n!}(2^{n+1}+n(n-1)-n(n-1)(n-2)-4n-4).$$
One can then check that $\f_3=\f_4=\f_5=0<\f_n$ for $n\geq 6$. Thus $\frac{d^2}{du^2}H(u)\geq 0$
implying that $H(u)$ is convex.
\end{proof}
\subsection{The case $m\sim n$}\label{m=n}
We will first prove Theorem \ref{min} under the assumption that $m=n$ and then in Sections \ref{gre} and \ref{le} we will extend
the result to arbitrary $m$. We will as usual prove that Hall's Condition holds \whp. We will therefore estimate the 
probability of the existence of sets $A,B$ where $|A|=k$ and $|B|\leq k-1$ such that $N_\G(A)\subseteq B$. Here $N_\G(S)$ 
is the set of neighbours of $S$ in $\G$. We call such a pair of sets, a {\em witness} to the non-existence of a perfect matching.
There are two possibilities to consider: (i) $A\subseteq L$ and $B\subseteq R$ or
(ii) $A\subseteq R$ and $B\subseteq L$. We deal with both cases in order to help extend the results to $m\neq n$. We observe
that if there exist a pair $A,B$ then there exist a minimal pair and in this case each $b\in B$ has at least two neighbours in $A$.
We deal first with the existence probability for a witness in Case (i) and leave Case (ii) until Section \ref{AinR}. We then 
combine these results to finish the case $m=n$ in Section \ref{combine}. We will deal computationally with {\em minimal witnesses} 
where each vertex in $B$ has at least 2 neighbours in $A$. If $v$ has a unique neighbour $w$ in $A$ then 
$A\setminus\set{w},B\setminus\set{v}$ is also a witness. 
\subsubsection{Case 1}\label{Case1}
We estimate
\begin{align}
&\p_L(k,\ell,D)=\nonumber\\
&\Pr(\exists A,B:\;|A|=k,|B|=\ell\leq \min\set{k-1,m/2},N_\G(A)=B,d(B)=D, d_A(b)\geq 2,b\in B)\leq\nonumber\\
&O(n^{1/2})\binom{n}{k}\binom{m}{\ell}\sum_{\substack{2\leq x_b\leq 
d_b,b\in[\ell]\\\sum_bx_b=kd\\\sum_{b\in [\ell]}d_b=D\\ \sum_{b\notin [\ell]}d_b=dn-D}}\prod_{b=1}^{m}\frac{z^{d_b}}{d_b!f(z)}
\binom{d_b}{x_b}\ (kd)!\prod_{i=0}^{dk-1}\frac{1}{dn-i}=\label{expl0}\\
&O(n^{1/2})\binom{n}{k}\binom{m}{\ell}\frac{(d(n-k))!}{(dn)!}\frac{(kd)!z^{dn}}{f(z)^{m}}\times\nonumber\\
&\hskip1in\brac{\sum_{\substack{2\leq x_b,b\in[{\ell}]\\\sum_bx_b=kd}}\prod_{b=1}^{\ell}\frac{1}{x_b!}}
\brac{\sum_{\substack{2\leq d_b,b\notin[{\ell}]\\\sum_bd_b=dn-D}}\prod_{b=k}^{m}\frac{1}{d_b!}}
\brac{\sum_{\substack{0\leq y_b,b\in[{\ell}]\\\sum_by_b=D-kd}}\prod_{b=1}^{\ell}\frac{1}{y_b!}}=\nonumber\\
&O(n^{1/2})\binom{n}{k}\binom{m}{\ell}\frac{(d(n-k))!}{(dn)!}\frac{(kd)!z^{dn}}{f(z)^{m}}\times\nonumber\\
&\hskip1in\brac{[u^{kd}](e^u-1-u)^{\ell}}\brac{[u^{dn-D}](e^u-1-u)^{m-\ell}}\brac{[u^{D-kd}]e^{u\ell}}\leq\label{3u}\\
&O(n^{1/2})\binom{n}{k}\binom{m}{\ell}\frac{(d(n-k))!}{(dn)!}\frac{(kd)!z^{dn}}{f(z)^{m}}
\frac{f(z)^{\ell}}{z^{kd}}\frac{f(\z_1)^{m-k+1}}{\z_1^{dn-D}}\frac{\ell^{D-dk}}{(D-kd)!}\leq\nonumber\\
\noalign{where $\z_1=\z(y)\leq z$ where $y=\frac{dn-D}{m-k+1}\geq 2$ due to our minimum degree assumption for $R$.} \nonumber\\
&O(n^{1/2})\binom{n}{k}\binom{m}{k-1}\frac{(d(n-k))!}{(dn)!}\frac{(kd)!z^{dn}}{f(z)^{m}}
\frac{f(z)^{k-1}}{z^{kd}}\frac{f(\z_1)^{m-k+1}}{\z_1^{dn-D}}\frac{(k-1)^{D-dk}}{(D-kd)!}\leq\label{4u}\\
&O\bfrac{k}{m^{1/2}}\frac{\binom{n}{k}\binom{m}{k}}{\binom{dn}{dk}}\brac{\frac{z^d}{f(z)^{\frac{m-k}{n-k}}}
\frac{f(\z_1)^{\frac{m-k}{n-k}}}{\z_1^{\frac{dn-D}{n-k}}}}^{n-k}\bfrac{ek}{D-dk}^{D-dk}.\nonumber\\
\noalign{Putting $k=an$ and $m=\b n$  and $h(u)=u^u(1-u)^{1-u}$ and $x=d-y$ where $0\le x \le d-2$ we obtain,
after substituting $\binom{n}{k}=O\bfrac{1}{k^{1/2}h(a)^n}$ etc. }
&\p_L(k,\ell,D)\leq O\bfrac{1}{n^{1/2}}\bfrac{h(a)^{d-1}}{h(a/\b)^{\b}}^n
\brac{\frac{z^d}{f(z)^{\frac{\b-a}{1-a}}}\frac{f(\z_1)^{\frac{\b-a}{1-a}}}{\z_1^{d-x}}
\bfrac{e\frac{a}{1-a}}{x}^x}^{n-k}.\label{1aa}
\end{align}

{\bf Explanation of \eqref{expl0}:}
Choose sets $A,B$ in  $\binom{n}{k}\binom{m}{\ell}$ ways. Choose degrees $d_b, b\in R$ with probability 
$O(n^{1/2})\prod_{b=1}^m\frac{z^{d_b}}{d_b!f(z)}$ such that $\sum_{b\in B}d_b=D,\,\sum_{b\notin B}d_b=dn-D$ for some $D\geq 2(\ell)$.
Choose the degrees $x_a,a\in A$ in the sub-graph induced by $A\cup B$. Having fixed the degree sequence,
we swap to the configuration model. Choose the configuration points associated
the $x_a,a\in A$ in $\prod_{a\in A}\binom{d}{x_a}$ ways. Assign these $D$ choices of points points associated with $A$ in $D!$ ways.
Then multiply by the probability $(kd)!\prod_{i=0}^{kd-1}\frac{1}{dn-i}$ of a given pairing of points in $A$.

{\bf Explanation of \eqref{3u} to \eqref{4u}:} If $A(x)=\sum_{n=0}^\infty a_nx^n$ where $a_n\geq 0$ for $n\geq 0$ 
then $a_n\leq A(\z)/\z^n$ for any positive $\z$ and $A(\z)/\z^n$ is minimised at $\z$ satisfying 
$\z A'(\z)/A(\z)=n$.

For the remainder of Section \ref{m=n} we assume that 
\beq{46a}
n\leq m\leq n+o(n^{7/8}).
\eeq
In which case we have
$$\bfrac{h(a)}{h(a/\b)^{\b}}^n\bfrac{f(\z_1)^{\frac{\b-1}{1-a}}}{f(z)^{\frac{\b-1}{1-a}}}^{n-k}=e^{o(n^{7/8}a)}.$$

Thus\eqref{1aa} becomes
\beq{1a}
\p_L(k,\ell,D)=O\bfrac{1}{n^{1/2}}e^{o(n^{7/8}a)}h(a)^{(d-2)n}
\brac{\frac{z^d}{f(z)}\frac{f(\z_1)}{\z_1^{d-x}}
\bfrac{e\frac{a}{1-a}}{x}^x}^{n-k}.
\eeq
{\bf Case 1.1:} $0\leq k\leq \brac{1-\frac{2}{d}}n$.

Observe (see \eqref{gz}) that
$$
\frac{z^d}{f(z)}\frac{f(\z_1)}{\z_1^{d-x}} = \frac{d^d}{g(d)}\frac{g(d-x)}{(d-x)^{d-x}}
$$
where $g(x)$ is as defined in Lemma \ref{use1}.

It follows from \eqref{gdash} that
\beq{int}
-\log\bfrac{g(d-x)}{g(d)} = \int_{d-x}^d\frac{d}{dt}\log(g(t))dt \ge \int_{d-x}^d(\log(1+\z e^{-\z})+1)dt.
\eeq
Now $\z e^{-\z}\leq e^{-1}$ which implies that $\log(1+\z e^{-\z})\geq \z e^{-\z}/10$. Also,
$$\frac{\z(t)}{t}=1-\frac{\z}{e^\z-1}\geq 1-\frac{2}{\z+2}=\frac{\z}{\z+2}.$$
And so $\z\leq t\leq \z+2$. Thus
$$\int_{d-x}^d(\log(1+\z e^{-\z})+1)dt\geq \int_{d-x}^d\brac{1+\frac{t-2}{10e^d}}dt=x+\frac{x(2d-4-x)}{20e^d}.$$
This implies that $\frac{g(d-x)}{g(d)}\le e^{-x}\psi(x)$ where $\psi(x)=e^{-\e_d(2d-4-x)x}$ and $\e_d=\frac{1}{20e^d}$. 
Note that $d-x\geq 2$ and so $\psi(x)\le e^{-(d-2)\e_dx}$ in the range of interest.
Plugging this into the last parenthesis of \eqref{1a} gives
\beq{1b}
\p_L(k,\ell,D)=O\bfrac{1}{n^{1/2}}e^{o(n^{7/8}a)}
\psi(x)^n\brac{h(a)^{d-2}\brac{d^d\frac{1}{(d-x)^{d-x}}\bfrac{\frac{a}{1-a}}{x}^x}^{1-a}}^n
\eeq
This immediately yields
\beq{A0}
A_{\ref{A0}}=\sum_{\ell<k=\e_Ln}^{n(1-2/d)}\sum_{D=dk}^{dk+n^{1/10}}
\p_L(k,\ell,D)\leq\sum_{\ell<k=\e_Ln}^{n(1-2/d)}\sum_{D=dk}^{dk+n^{1/10}}
O\bfrac{\log n}{n^{1/2}}h(a)^{(d-2)n}e^{o(n)}=o(1).
\eeq
We use the notation $A_{\ref{A0}}$ so that the reader can easily refer back to the equation giving its definition.

We will work with $D\leq k\log n$ because it is easy to show that \whp\ the maximum degree in $\G$ is $o(\log n)$.
The bound for $A_{\ref{A0}}$ comes 
from \eqref{1b}, using the fact that $h(a)$ is bounded away from 1 and $x=o(1)$ in this summation.
$A_{\ref{A0}}$ is the first of several sums that 
together show the unlikelihood chance of a witness. We will display them as they
become available and use them in Sections \ref{combine}, \ref{gre} and \ref{le}.

The main term $h(a)^{d-2}\brac{d^d\frac{1}{(d-x)^{d-x}}\bfrac{\frac{a}{1-a}}{x}^x}^{1-a}$ 
in \eqref{1b} is maximized when $x = ad$, provided $ad\leq d-2$ or $k\leq n\brac{1-\frac{2}{d}}$. This in turn gives
\begin{align}
\p_L(k,D)&=O\bfrac{1}{n^{1/2}}e^{o(n^{7/8}a)}\psi(x)^n
\brac{h(a)^{d-2}\brac{d^d\frac{1}{(d-ad)^{d-ad}}\frac{1}{((1-a)d)^{ad}}}^{1-a}}^n\nonumber\\
&=O\bfrac{1}{n^{1/2}}e^{o(n^{7/8}a)}\psi(x)^n
\brac{h(a)^{d-2}\brac{\frac{d^{d}}{d^{d-ad}d^{ad}}(1-a)^{-d}}^{1-a}}^n\nonumber\\
&\leq O\bfrac{1}{n^{1/2}}e^{o(n^{7/8}a)}\psi(x)^n\brac{a^{a(d-2)}(1-a)^{-2(1-a)}}^n\label{42}
\end{align}

The function $\r_d(a)=a^{a(d-2)}(1-a)^{-2(1-a)}$ is at most 1 and is log-convex in $a$ on $[0,1-\frac{2}{d}]$. 
Indeed, if $L_1(a)=\log\r_d(a)$ then
\begin{align}
&\frac{dL_1}{da}=d-2+(d-2)\log a+2\log(1-a)\label{k1}\\
&\frac{d^2L_1}{da^2}=\frac{d-2-da}{a(1-a)}\label{k2}
\end{align}
We have $L_1(0)=0$ and $L_1'(0)=-\infty$. It follows that for every $K>0$ there exists a constant $\e_L(K,d)>0$ such that 
\beq{43}
\r_d(a)\leq e^{-Ka}\qquad for\ a\leq \e_L(K).
\eeq 
We let $\e_L=\e_L(1,d)$.

We can immediately write
\beq{A1}
A_{\ref{A1}}=\sum_{\ell<k=2}^{n^{1/10}}\sum_{D=dk}^{k\log n}\p_L(k,\ell,D)=
\sum_{\ell<k=2}^{n^{1/10}}\sum_{D=dk}^{\log n}O\bfrac{\log n}{n^{1/2}}e^{o(kn^{-1/8})}=o(1).
\eeq
The bound for $A_{\ref{A1}}$ is derived from \eqref{42} using $\psi(x),\r(a)\leq 1$. 

Along the same lines we have
\beq{A2}
A_{\ref{A2}}=\sum_{\ell<k=n^{1/10}}^{\e_Ln}\sum_{D=dk}^{k\log n}\p_L(k,\ell,D)=
\sum_{\ell<k=n^{1/10}}^{\e_Ln}\sum_{D=dk}^{k\log n}O\bfrac{\log n}{n^{1/2}}
e^{-k(1-o(n^{-1/8})}=o(1).
\eeq
The bound for $A_{\ref{A2}}$
comes from \eqref{42} and \eqref{43}. 

Now
$\brac{1-\frac{2}{d}}^{(d-2)^2/d}\bfrac{2}{d}^{-4/d}$ decreases in $d$ and is $\le 1$ for $d\ge 4$.
So if $d\geq 4$ then
\beq{A3}
A_{\ref{A3}}=\sum_{\ell<k=\e_Ln}^{n(1-2/d)}\sum_{D=dk+n^{1/10}}^{\log n}\p_L(k,\ell,D)=
\sum_{\ell<k=\e_Ln}^{n(1-2/d)}\sum_{D=dk+n^{1/10}}^{k\log n}
O\bfrac{\log n}{n^{1/2}}e^{o(n^{7/8}a)}\psi(n^{-4/5})^n\r_d(a)^n=o(1).
\eeq
The bound for $A_{\ref{A3}}$ comes from \eqref{42} using the fact that $\r_d(a)\leq e^{-a}$ and $x\geq n^{-4/5}$ in this summation. 

When $d=3$ we need some extra calculations. First note that $\r_3(.15)<1$ and so arguing as above we have
\beq{A33}
A_{\ref{A33}}=\sum_{\ell<k=\e_Ln}^{.15n}\sum_{D=3k+n^{1/10}}^{k\log n}\p_L(k,\ell,D)=
\sum_{\ell<k=\e_Ln}^{.15n}\sum_{D=3k+n^{1/10}}^{\log n}O\bfrac{\log n}{n^{1/2}}e^{o(n^{7/8}a)}\psi(n^{-4/5})^n\r_3(a)^n=o(1).
\eeq
Because we can choose any value for $\z_1$ in the bound \eqref{1a} we can simplify matters by choosing $\z_1=\xi$ 
independent of $x$ to get
\beq{1a1a}
\p_L(k,D)=O\bfrac{1}{n^{1/2}}e^{o(n^{7/8}a)}h(a)^{n}
\brac{\frac{z^3}{f(z)}\frac{f(\xi)}{\xi^{3}}
\bfrac{e\xi a}{x(1-a)}^x}^{n-k}.
\eeq
Now 
\beq{xxq}
\bfrac{\xi ea}{(1-a)x}^x\leq \exp\set{\frac{\xi a}{1-a}}
\eeq
and so
\beq{are}
\p_L(k,D)\leq O\bfrac{k}{n^{1/2}}\brac{h(a)e^{\xi a}e^{o(n^{-1/8}a)}\brac{\frac{z^3}{f(z)}\frac{f(\xi)}{\xi^3}}^{1-a}}^n.
\eeq
Now the function $L_2(a)=h(a)e^{\xi a}\brac{\frac{z^3}{f(z)}\frac{f(\xi)}{\xi^3}}^{1-a}$ is log-convex. Our choice of
$\xi$ will be 1.5 and we note that with this choice
$L_2(.15),L_2(2/5)<.98$ and so
\beq{A333}
A_{\ref{A333}}=\sum_{\ell<k=.15n}^{2n/5}\sum_{D=3k+n^{1/10}}^{k\log n}\p_L(k,\ell,D)\leq
\sum_{\ell<k=.15n}^{2n/5}\sum_{D=3k+n^{1/10}}^{k\log n}O\bfrac{\log n}{n^{1/2}}e^{o(n^{7/8}a)}(.98)^n=o(1).
\eeq
We have gone slightly beyond $n/3$ to $2n/5$. It is convenient to repeat this idea for a couple of ranges. Putting $\xi=.5$
we get $L_2(2/5),L_2(.74)<.995$ from which we deuce that
\beq{A3333}
A_{\ref{A3333}}=\sum_{\ell<k=2n/5}^{.74n}\sum_{D=3k+n^{1/10}}^{k\log n}\p_L(k,\ell,D)\leq
\sum_{\ell<k=2n/5}^{.74n}\sum_{D=3k+n^{1/10}}^{k\log n}O\bfrac{\log n}{n^{1/2}}e^{o(n^{7/8}a)}(.995)^n=o(1).
\eeq
Putting $\xi=.2$
we get $L_2(.74),L_2(.87)<.995$ from which we deuce that
\beq{bbx}
A_{\ref{bbx}}=\sum_{\ell<k=.74}^{.87n}\sum_{D=3k+n^{1/10}}^{k\log n}\p_L(k,\ell,D)\leq
\sum_{\ell<k=.74}^{.87n}\sum_{D=3k+n^{1/10}}^{k\log n}O\bfrac{\log n}{n^{1/2}}e^{o(n^{7/8}a)}(.995)^n=o(1).
\eeq
{\bf Case 1.2.1:} $\brac{1-\frac{2}{d}}n\leq k\leq \brac{1-\frac{1}{d-1}}n$.

For $a\ge 1-\frac{2}{d}$ the maximising value for $x$ in \eqref{1b} is 
at $x = d-2$ (recall that $0\le x\le d-2)$, so plugging into \eqref{1b} gives

\begin{align}
\p_L(k,\ell,D)&=O\bfrac{k}{n^{1/2}}\brac{e^{o(n^{-1/8}a)}h(a)^{d-2} \brac{d^d\frac{1}{2^2}\bfrac{\frac{a}{1-a}}{d-2}^{d-2}}^{1-a}}^n\nonumber\\
&=O\bfrac{k}{n^{1/2}}\brac{e^{o(n^{-1/8}a)}a^{d-2}\bfrac{d^d}{(d-2)^{d-2}2^2}^{1-a}}^n\label{2a}
\end{align}

Let $L_3(a)=\log\brac{a^{d-2}\bfrac{d^d}{(d-2)^{d-2}2^2}^{1-a}}$. Then
\begin{align*}
\frac{d}{d a} L_3(a) &= \frac{d}{da}\brac{(d-2)\log a + (1-a)\log\bfrac{d^d}{(d-2)^{d-2}2^2}}\\
&=\frac{d-2}{a} -d\log d + (d-2)\log(d-2)+2\log2
\end{align*}
Assume for now that $d\geq 6$.
Then the derivative with respect to $d$, for of the last expression is
$$
\frac{1}{a}-\log\bfrac{d}{d-2} \ge 1-\log\bfrac{6}{4}>0
$$
so it takes a minimum at $d=6$ with a value
$$
\frac{4}{a}-6\log 6 +4\log 4+ 2\log 2 \ge 0.18>0.
$$
Now $L_3(1)=0$ and so for $a\geq 1-\frac{2}{d}$ and $d\ge 6$ we have
\beq{A4}
A_{\ref{A4}}=\sum_{\ell<k=n(1-2/d)}^{n-n^{7/8}}\sum_{D=dk}^{k\log n}\p_L(k,\ell,D)
\leq \sum_{\ell<k=n(1-2/d)}^{n-n^{7/8}}\sum_{D=dk}^{k\log n}
O\bfrac{\log n}{n^{1/2}}e^{-.18(n-k)+o(n^{-1/8}k)}=o(1).
\eeq

For $d=3,4,5$ we use the following
\begin{claim}\label{claim:fz}
For $y\ge 2$ we have $\frac{f(\z(y))}{\z(y)^y} \le \frac{3}{4}$.
\end{claim}

Substituting this into \eqref{1a} gives
\beq{2b}
\p_L(k,\ell,D)=O\bfrac{1}{n^{1/2}}\brac{h(a)^{d-2}
\brac{e^{o(n^{-1/8}a)}\frac{z^d}{f(z)}\frac{3}{4}\bfrac{e\frac{a}{1-a}}{x}^x}^{1-a}}^n
\eeq

The maximum of $L_3(x)=h(a)^{d-2}\brac{\frac{z^d}{f(z)}\frac{3}{4}\bfrac{e\frac{a}{1-a}}{x}^x}^{1-a}$ 
is taken when either $x = \frac{a}{1-a}$ for $a \in [1-\frac{2}{d},1-\frac{1}{d-1}]$ or at $x=d-2$ 
for $a\in[1-\frac{1}{d-1},1]$.

So for $a\in[1-\frac{2}{d},1-\frac{1}{d-1}]$ we get
$$
\p_L(k,\ell,D)\leq O\bfrac{1}{n^{1/2}}\brac{e^{o(n^{-1/8}a)}h(a)^{d-2}e^a\bfrac{3z^d}{4f(z)}^{1-a}}^n.
$$
The expression $L_4(a)=h(a)^{d-2}e^a\bfrac{3z^d}{4f(z)}^{1-a}$ is log-convex on $[1-\frac{2}{d},1-\frac{1}{d-1}]$ 
and $L_4\leq .97$ at both ends of the interval for both $d=3,4,5$.
We can therefore write
\beq{A4'}
A_{\ref{A4'}}=\sum_{\ell<k=n(1-2/d)}^{n(1-1/(d-1))}\sum_{D=dk}^{k\log n}\p_L(k,\ell,D)
\leq \sum_{\ell<k=n(1-2/d)}^{n(1-1/(d-1))}\sum_{D=dk}^{k\log n}
O\bfrac{\log n}{n^{1/2}}e^{o(n^{7/8}a)}(.97)^n=o(1).
\eeq
for $a\in [1-\frac{2}{d},1-\frac{1}{d-1}]$.

{\bf Case 1.2.2:} $k\geq \brac{1-\frac{1}{d-1}}n$.

For $a\in[1-\frac{1}{d-1},1]$ from \eqref{2b} we get
\begin{align*}
\p_L(k,\ell,D)&=O\bfrac{1}{n^{1/2}}\brac{e^{o(n^{-1/8}a)}
h(a)^{d-2}e^a\bfrac{3z^d}{4f(z)}^{1-a}\bfrac{e\frac{a}{1-a}}{d-2}^{(1-a)(d-2)}}^n\nonumber\\
&=O\bfrac{1}{n^{1/2}}\brac{e^{o(n^{-1/8}a)}a^{d-2}\bfrac{3z^de^{d-2}}{4f(z)(d-2)^{d-2}}^{1-a}}^n
\end{align*}
The expression $L_5(a)=a^{d-2}\bfrac{3z^de^{d-2}}{4f(z)(d-2)^{d-2}}^{1-a}$ 
is log-concave on $[1-\frac{1}{d-1},1]$.  
The derivative of $\log L_5$ at $a=1$ is at least 1/100 for both $d=4,5$. Consequently for $d=4,5$
\beq{A4''}
A_{\ref{A4''}}=\sum_{\ell<k=n(1-1/(d-1))}^{n-n^{7/8}}\sum_{D=dk}^{k\log n}\p_L(k,\ell,D)\leq \sum_{\ell<k=n(1-2/d)}^{n-n^{7/8}}\sum_{D=dk}^{k\log n}
O\bfrac{\log n}{n^{1/2}}e^{-(n-k)/100+o(n^{-1/8}k)}=o(1).
\eeq
for $a\in [1-\frac{1}{d-1},1]$.

For $d=3$ we go back to \eqref{1a1a} and \eqref{xxq} and put $\xi=(1-a)/a$ giving
\beq{less1}
\p_L(k,D)\leq O\bfrac{1}{n^{1/2}}\brac{e^{o(n^{-1/8}a)}h(a)\brac{\frac{z^3 }{f(z)}\frac{ea^3f((1-a)/a)}{(1-a)^3}}^{1-a}}^n.
\eeq
Now for $x<1$ we have
$$f(x)\leq \frac{x^2}{2}\brac{1+\frac{x}{3}}.$$
Plugging this into \eqref{less1} for $a\geq 1/2$ and replacing $2a+1\leq3$ we have
\beq{less2}
\p_L(k,\ell,D)\leq O\bfrac{1}{n^{1/2}}\brac{e^{o(n^{-1/8}a)}a^a\brac{\frac{ez^3}{2f(z)}}^{1-a}}^n.
\eeq
The function $L_6(a)=a^a\brac{\frac{ez^3}{2f(z)}}^{1-a}$ is log-convex and $L_6(.84)<.9995$ and $L_6(1)=1$. 
Also, $L_6'(1)>1/20$. It follows that
\beq{A4'''}
A_{\ref{A4'''}}=\sum_{\ell<k=.84n}^{n-n^{7/8}}\sum_{D=dk}^{k\log n}\p_L(k,\ell,D)\leq \sum_{\ell<k=.84n}^{n-n^{7/8}}\sum_{D=dk}^{k\log n}
O\bfrac{k\log n}{n^{1/2}}e^{-\min\set{(n-k)/20,n/1000}+o(n^{-1/8}k)}=o(1).
\eeq
\begin{proofof}{Claim \ref{claim:fz}}
Recall from Lemma \ref{use2} that $\z(x)$ is concave and thus $\frac{\z(y)-\z(2)}{y-2}$
is decreasing. Since $\z(2) = 0$ we have
\begin{align*}
\frac{\z(y)}{y-2} &= \frac{\z(y)-\z(2)}{y-2} \\
&\le \lim_{y\to 2}\frac{d}{dy}\z(y)\\
&= \lim_{z\to 0}\frac{e^z-z-1}{(e^z-1)^2-z^2e^z}\\
&= 3
\end{align*}
Thus we have that $\z(y)\le 3(y-2)$ for all $y$. For $y\ge 3$ we can upper bound
$$
\frac{f(\z(y))}{\z(y)^y} \le \frac{f(y)}{y^y} \le \frac{e^y}{y^y} \le \frac{e^3}{3^3} \le \frac{3}{4}
$$
For $y\in[2,3]$ we have that $3(y-2)\le y$ and so we can bound
$$
\frac{f(\z(y))}{\z(y)^y} \le \frac{f(3(y-2))}{(3(y-2))^y}
$$
Taking the logarithm of this expression and substituting $u = 3(y-2)$ we get
\beq{2bb}
\log f(u)-\brac{\frac{u}{3}+2}\log u = u - \frac{u}{3}\log(u) + \brac{\log f(u)-u-2\log(u)}=u - \frac{u}{3}\log(u) + H(u)
\eeq
where $H(u)$ is from Lemma \ref{use3}.

Since $H(u)$ is convex we have
$$
H(u) \le H(0)+u\frac{H(3)-H(0)}{3}
$$
Pluggin this into \eqref{2bb} we get
$$
u-\frac{u}{3}\log{u}+H(0)+u\frac{H(3)-H(0)}{3} = H(0)+\frac{1}{3}\brac{(H(3)-H(0)+3)u-u\log u}
$$
which is concave in $u$ and takes a maximum value of $H(0)+u/3$ when
$$
u=\exp(H(3)-H(0)+3-1) = \frac{e^3-4}{18e} \le 0.33
$$
Pluggin this back in we see that for $x\in [2,3]$, which is $u\in [0,3]$ we have
$$
\frac{f(\z(y))}{\z(y)^y} \le \exp(H(0)+.11) = \exp(0.11)/2 < \frac{3}{4} 
$$
\end{proofof}

\subsubsection{Case 2}\label{AinR}
Now let us estimate the probability of a violation of Hall's condition with $A\subseteq R$. We once again begin with arbitrary $m$.
Let
\begin{align}
&\p_R(k,\ell,D)=\nonumber\\
&\Pr(\exists A\subseteq R,B\subseteq L:\;|A|=k,|B|=\ell\leq \min\set{k-1,n/2},N_\G(A)
\subseteq B,d_B(b)\geq 2,b\in B,d_R(A)=D)\leq \nonumber\\
&O(n^{1/2})\binom{m}{k}\binom{n}{\ell}\sum_{\substack{2\leq  d_a,a\in[m]\\2\leq x_b\leq d,b\in[\ell]\\
\sum_{a\in [k]}d_a=\sum_{b\in [\ell]}x_b=D\\\sum_{a\notin [k]}d_a=dn-D}}\prod_{a=1}^m\frac{z^{d_a}}{d_a!f(z)}
\prod_{b=1}^{\ell}\binom{d}{x_b}\ D!\prod_{i=0}^{D-1}\frac{1}{dn-i}=\label{expl1}\\
&O(n^{1/2})\binom{m}{k}\binom{n}{\ell}\frac{z^{dn}D!}{f(z)^m}\frac{(dn-D)!}{(dn)!}\nonumber\\
&\hskip3cm\times\brac{\sum_{\substack{2\leq  d_a,a\in[k]\\\sum_ad_a=D}}\prod_{a=1}^k\frac{1}{d_a!}}
\brac{\sum_{\substack{2\leq  d_a,a\notin[k]\\\sum_ad_a=dn-D}}\prod_{a=k+1}^n\frac{1}{d_a!}}
\brac{\sum_{\substack{2\leq x_b\leq d,b\in[\ell]\\\sum_bx_b=D}}\prod_{b=1}^{\ell}\binom{d}{x_b}}=\nonumber\\
&O(n^{1/2})\binom{m}{k}\binom{n}{\ell}\frac{z^{dn}}{f(z)^m}\frac{1}{\binom{dn}{D}}\nonumber\\
&\hskip1cm\times\brac{[u^D](e^u-1-u)^k}\brac{[u^{dn-D}](e^u-1-u)^{m-k}}\brac{[u^D]((1+u)^d-(1+du))^\ell}\leq\nonumber\\
&O(n^{1/2})\binom{m}{k}\binom{n}{k-1}\frac{z^{dn}}{f(z)^m}\frac{1}{\binom{dn}{D}}\nonumber\\
&\hskip1cm\times\brac{[u^D](e^u-1-u)^k}\brac{[u^{dn-D}](e^u-1-u)^{m-k}}\brac{[u^D]((1+u)^d-(1+du))^k}\leq\label{colbert}\\
&O(n^{1/2})\binom{m}{k}\binom{n}{k-1}\frac{z^{dn}}{f(z)^m}\frac{1}{\binom{dn}{D}}\frac{f(\z_1)^k}{\z_1^D}
\frac{f(\z_2)^{m-k}}{\z_2^{dn-D}}\binom{dk}{D}\leq\nonumber\\
\noalign{where $\z_1=\z(D/k)$ and
$\z_2=\z\bfrac{dn-D}{m-k}$ -- actually any value for $\z_1,\z_2$ is valid --}\nonumber\\
&O(n^{1/2})\binom{m}{k}\binom{n}{k}\frac{\binom{dk}{D}}{\binom{dn}{D}}
\bfrac{f(\z_1)}{f(z)}^k\bfrac{f(\z_2)}{f(z)}^{m-k}\bfrac{z}{\z_1}^D\bfrac{z}{\z_2}^{dn-D}\label{3xx}\\
&=O\bfrac{1}{n^{1/2}}\brac{\frac{h(\th a/d)^d}{h(a)h(a/\b)^\b h(\th/d)^{ad}}\bfrac{f(\z_1)}{f(z)}^a\bfrac{f(\z_2)}{f(z)}^{\b-a}
\bfrac{z}{\z_1}^{\th a}\bfrac{z}{\z_2}^{d-\th a}}^n\label{3axx}\\
&=O\bfrac{1}{n^{1/2}}\brac{\frac{h(\th a/d)^d}{h(a)h(a/\b)^\b h(\th/d)^{ad}}
\frac{z^d}{f(z)^\b}\frac{f(\z_1)^a}{\z_1^{\th a}}\frac{f(\z_2)^{\b-a}}{\z_2^{d-\th a}}}^n\label{3ayxw}
\end{align}
where $a=k/n$, $m=\b n$ and $D=\th k\leq dk$.

{\bf Explanation of \eqref{expl1}:}
Choose sets $A,B$ in  $\binom{m}{k}\binom{n}{\ell}$ ways. Choose degrees $d_a, a\in R$ with probability 
$O(n^{1/2})\prod_{a=1}^n\frac{z^{d_a}}{d_a!f(z)}$ such that $\sum_{a\in A}d_a=D,\,\sum_{a\notin A}d_a=dn-D$ for some $D\geq 2k$.
Choose the degrees $x_b,b\in B$ in the sub-graph induced by $A\cup B$. Having fixed the degree sequence,
swap to the configuration model \cite{BoCo}. Choose the configuration points associated with
the $x_b,b\in B$ in $\prod_{b\in B}\binom{d}{x_b}$ ways. 
Then multiply by the probability $D!\prod_{i=0}^{D-1}\frac{1}{dn-i}$ of a given pairing of points in $A$.

We assume that \eqref{46a} holds
for the remainder of the section. In which case we have
$$\frac{h(a)}{h(a/\b)^{\b}}\frac{f(\z_2)^{\b-1}}{f(z)^{\b-1}}=e^{o(n^{-1/8}a)}.$$
Thus,
\eqref{3ayxw} becomes
\beq{3ayx}
\p_R(k,\ell,D)\leq O\bfrac{1}{n^{1/2}}\brac{e^{o(n^{-1/8}a)}
\frac{h(\th a/d)^d}{h(a)^2h(\th/d)^{ad}}\frac{z^d}{f(z)}\frac{f(\z_1)^a}{\z_1^{\th a}}
\frac{f(\z_2)^{1-a}}{\z_2^{d-\th a}}}^n
\eeq

\ignore{{\bf {\Large Mathematica says that $F_\th(a)=\frac{h(\th a/3)^33^{(3-\th)a}}{h(a)^2h(\th-2)^a}$ works for $d=3$,
provided $2.08\leq\th\leq 3$.\\
Note that
$\frac{z^d}{f(z)}\frac{f(\z_1)^a}{\z_1^{\th a}}\frac{f(\z_2)^{1-a}}{\z_2^{d-\th a}}\leq 1$.\\
For a fixed $\th$ the function $F_\th(a)$ is log-concave for $0\leq a\leq a_\th=\frac{3(\th-2)}{\th}$ and log-convex for
$a\geq a_\th$. The function $F\th(a_\th)<1$ for $\th\geq 2.08$.
}}}

It follows from Lemma \ref{use1} that we can upper bound
\begin{align*}
\frac{z^d}{f(z)}\bfrac{f(\z_1)}{\z_1^\th}^a\bfrac{f(\z_2)}{\z_2^{\frac{d-a\th}{1-a}}}^{1-a}
&= \frac{d^d}{g(d)}\frac{g(\th)^a g\bfrac{d-a\th}{1-a}^{1-a}}{\th^{a\th}\bfrac{d-a\th}{1-a}^{d-a\th}}\\
&\le \frac{g(a\th+(1-a)\frac{d-a\th}{1-a})}{g(d)} \frac{d^d}{\th^{a\th}\bfrac{d-a\th}{1-a}^{d-a\th}}\\
&= \frac{a^{a\th}(1-a)^{d-a\th}}{\brac{\bfrac{a\th}{d}^{\frac{a\th}{d}}(1-\frac{a\th}{d})^{1-\frac{a\th}{d}}}^d}\\
&= \frac{a^{a\th}(1-a)^{d-a\th}}{h\bfrac{a\th}{d}^d}
\end{align*}

Plugging this into \eqref{3ayx} gives

\beq{4x}
\p_R(k,\ell,D)\leq O\bfrac{1}{n^{1/2}}\bfrac{e^{o(n^{-1/8}a)}a^{a\th}(1-a)^{d-a\th}}{h(a)^2h\bfrac{\th}{d}^{ad}}^n
\eeq
Now let $R_1(\th)=\log\bfrac{a^{a\th}(1-a)^{d-a\th}}{h(a)^2h\bfrac{\th}{d}^{ad}}$. Then
\begin{align*}
&R_1'(\th)=a\log a - a\log (1-a)-a\log\th+a\log(d-\th).\\
&R_1''(\th)=-\frac{ad}{\th(d-\th)}<0.
\end{align*}

Thus $R_1(\th)$ is concave and is maximized when $\th=ad$. Because $\th\geq 2$ we can only use this for 
$a\geq 2/d$. 

{\bf Case 2.1:} $k\geq 2n/d$.

\begin{align}
\p_R(k,\ell,D)&\leq O\bfrac{1}{n^{1/2}}\bfrac{e^{o(n^{-1/8}a)}a^{a^2d}(1-a)^{d-a^2d}}{h(a)^{2+ad}}^n\nonumber\\
&= O\bfrac{1}{n^{1/2}}\brac{e^{o(n^{-1/8}a)}a^{a^2d-a(2+ad)}(1-a)^{d-a^2d-(1-a)(2+ad)}}^n\nonumber\\
&= O\bfrac{1}{n^{1/2}}\brac{e^{o(n^{-1/8}a)}a^{-2a}(1-a)^{(d-2)(1-a)}}^n\nonumber\\
&= O\bfrac{1}{n^{1/2}}e^{o(n^{7/8}a)}\r_d(1-a)^n\label{1-a}
\end{align}
where the function $\r_d$ is defined following \eqref{42}.

We find that 
\beq{d=5}
\r_d(1-2/d)=\brac{\frac{d^4}{16}\brac{1-\frac{2}{d}}^{(d-2)^2}}^{1/d}\leq .9\text{ for }d\geq 5.
\eeq
Now $\r_d(1-2/d)<9/10$ for $d\geq 5$ and $\r_4(2.01/4)<.997$. So, with the aid of \eqref{43},
\begin{align}
&B_{\ref{B0}}=\sum_{\ell<k=2n/d}^{n-n^{7/8}}
\sum_{D=2k}^{k\log n}\p_R(k,\ell,D)\leq \sum_{\ell<k=2n/d}^{n-n^{7/8}}O\bfrac{1}{n^{1/2}}e^{o(n^{7/8}a)}
\sum_{D=2k}^{k\log n}e^{-n^{7/8}}
\qquad for\ d\geq 5.\label{B0}\\
&B_{\ref{B00}}=\sum_{\ell<k=2.01n/4}^{n-n^{7/8}}
\sum_{D=2k}^{dk}\p_R(k,\ell,D)\leq \sum_{\ell<k=2.01n/4}^{n-n^{7/8}}O\bfrac{1}{n^{1/2}}e^{o(n^{7/8}a)}
\sum_{D=2k}^{dk}e^{-n^{7/8}}
\qquad for\ d=4.\label{B00}
\end{align}
We will treat $d=3$ and $k\geq 2n/3$ under Case 2.2.

{\bf Case 2.2:} $2\leq k\leq 2n/d$.

In this case the expression in \eqref{4x} (ignoring error terms) is maximized at $\th=2$. Then
\begin{align*}
\p_R(k,D)&\leq O\bfrac{1}{n^{1/2}}\bfrac{e^{o(n^{-1/8}a)}a^{2a}(1-a)^{d-2a}}{h(a)^2h\bfrac{2}{d}^{ad}}^n\\
&= O\bfrac{1}{n^{1/2}}\bfrac{e^{o(n^{-1/8}a)}(1-a)^{d-2a-2(1-a)}}{h\bfrac{2}{d}^{ad}}^n\\
&= O\bfrac{1}{n^{1/2}}\bfrac{e^{o(n^{-1/8}a)}(1-a)^{d-2}}{h\bfrac{2}{d}^{ad}}^n
\end{align*}
Let $R_2(a)=\log\bfrac{(1-a)^{d-2}}{h\bfrac{2}{d}^{ad}}$. Then
\begin{align*}
&R_2'(a)=-\frac{d-2}{1-a}-d\log h(2/d)<0\qquad for\ d\geq 6.\\
&R_2''(a)=-\frac{d-2}{(1-a)^2}<0.
\end{align*}

Thus $R_2(a)$ is strict concave and its maximum is taken at $a=0$ and $R_2(a)\leq R_2(0)a$ for all $a\in[0,\frac{2}{d}]$.
Furthermore, $R_2(0)<-3/10$ for $d\geq 6$. It follows that if $d\geq 6$ then
\begin{eqnarray}
B_{\ref{B1}}&=&\sum_{\ell<k=2}^{n^{1/10}}\sum_{D=2k}^{dk}
\p_R(k,\ell,D)\leq \sum_{\ell<k=2}^{n^{1/10}}\sum_{D=2k}^{dk}O\bfrac{1}{n^{1/2}}e^{o(n^{-1/8}k)}=o(1).\label{B1}\\
B_{\ref{B2}}&=&\sum_{\ell<k=n^{1/10}}^{2n/d}\sum_{D=2k}^{dk}
\p_R(k,\ell,D)\leq \sum_{\ell<k=n^{1/10}}^{2n/d}\sum_{D=2k}^{dk}O\bfrac{1}{n^{1/2}}e^{-(3/10+o(n^{-1/8})k}=o(1).
\label{B2}
\end{eqnarray}
For $d=3,4,5$ we use a better bound on $[u^D]((1+u)^d-1-du)^k$ in \eqref{colbert}. 

{\bf Case 2.2a: $d=5$.}
\begin{align*}
[u^D]((1+u)^5-1-5u)^k &= [u^D](10u^2 + 10u^2+5u^4+u^5)^k)\\
&= [u^{D-2k}](10+10u+5u^2+u^3)^k\\
&= 10^k[u^{D-2k}]\brac{1+u+\frac{u^2}{2}+\frac{u^3}{10}}^k\\
&\le 10^k[u^{D-2k}]\brac{1+\frac{u}{2}}^{3k}\\
&= 10^k \frac{\binom{3k}{D-2k}}{2^{D-2k}}
\end{align*}

Replacing the $\frac{1}{h\bfrac{\th}{d}^{ad}}$ factor in \eqref{4x} which comes from $\binom{dk}{D}$ gives, for d=5,
\begin{align}
\p_R(k,D)&\leq O\bfrac{k}{n^{1/2}}\bfrac{e^{o(n^{-1/8}a)}
a^{a\th}(1-a)^{5-a\th}}{h(a)^2}^n\bfrac{10}{2^{\th-2} h\bfrac{\th-2}{3}^3}^k\nonumber\\
&=O\bfrac{k}{n^{1/2}}\bfrac{e^{o(n^{-1/8}a)}10^a a^{a(\th-2)}(1-a)^{5-2-a(\th-2)}}{\brac{2^{\th-2}h\bfrac{\th-2}{3}^3}^a}^n\nonumber\\
&=O\bfrac{k}{n^{1/2}}\brac{e^{o(n^{-1/8}a)}10^a(1-a)^{3}
\bfrac{\bfrac{a}{1-a}^{\frac{\th-2}{3}}}{2^{\frac{\th-2}{3}}h\bfrac{\th-2}{3}}^{3a}}^n\label{6b}
\end{align}

Let $p(x) = \frac{q^x}{h(x)}$ for any $x\in[0,1]$, note that if $P(x)=\log p(x)$ then
\begin{align*}
&P'(x) = \log q - \log x +\log(1-x)\\
&P''(x)=-\frac{1}{x}-\frac{1}{1-x}<0
\end{align*}
and so $p(x)$ is maximized when $\log q = \log\bfrac{x}{1-x}$ or $x = \frac{q}{1+q}$ and the maximum value is $1+q$

Thus from \eqref{6b} we get 
\begin{align*}
\p_R(k,\ell,D)\leq O\bfrac{k}{n^{1/2}}\brac{e^{o(n^{-1/8}a)}10^a(1-a)^{3}\brac{1+\frac{a}{2(1-a)}}^{3a}}^n
\end{align*}
Let $R_3(a)=\log\brac{10^a(1-a)^{3}\brac{1+\frac{a}{2(1-a)}}^{3a}}$. Then
\begin{align*}
&R_3'(a)=\log 10-\frac{6}{2-a}+3\log\bfrac{2-a}{2-2a}\\
&R_3''(a)=\frac{3a}{(2-a)^2(1-a)}>0.
\end{align*}

So $R_3(a)$ is log-convex on $[0,\frac{2}{5}]$. We have $R_3(0)=0$ and $R_3'(0)=\log 10 -3\leq -3/4$
and $R_3(2/5)<-1/4$. It follows that
\begin{multline}\label{B3}
B_{\ref{B3}}=\sum_{\ell<k=2}^{2n/5}\sum_{D=2k}^{5k}
\p_R(k,\ell,D)\leq \\
\sum_{\ell<k=2}^{n^{1/10}}\sum_{D=2k}^{5k}O\bfrac{1}{n^{1/2}}e^{o(n^{-1/8}a)}+
\sum_{\ell<k=n^{1/10}}^{2n/5}\sum_{D=2k}^{5k}O\bfrac{1}{n^{1/2}}e^{-(3/4+o(n^{-7/8}))k}=o(1).
\end{multline}
{\bf Case 2.2b: $d=4$.}
\begin{align*}
[u^D]((1+u)^4-1-4u)^k &= [u^{D-2k}](6+4u+u^2)^k\\
&= 6^k[u^{D-2k}]\brac{1+\frac{4}{6}u+\frac{u^2}{6}}^k\\
&\le 6^k[u^{D-2k}]\brac{1+\frac{u}{2}}^k\\
&= 6^k\frac{\binom{2k}{D-2k}}{2^{D-2k}}.
\end{align*}

\begin{align*}
\p_R(k,\ell,D)&\leq O\bfrac{1}{n^{1/2}}\bfrac{e^{o(n^{-1/8}a)}a^{a\th}(1-a)^{4-a\th}}{h(a)^2}^n\bfrac{6}{2^{\th-2}h\bfrac{\th-2}{2}^2}^k\\
&=O\bfrac{1}{n^{1/2}}\brac{e^{o(n^{-1/8}a)}6^aa^{a(\th-2)}(1-a)^{2-a(\th-2)}\bfrac{1}{2^{\frac{\th-2}{2}}h\bfrac{\th-2}{2}}^{2a}}^n\\
&= O\bfrac{1}{n^{1/2}}\brac{e^{o(n^{-1/8}a)}6^a(1-a)^2\bfrac{\bfrac{a}{1-a}^{\frac{\th-2}{2}}}{2^{\frac{\th-2}{2}}h\bfrac{\th-2}{2}}^{2a}}^n\\
&\le O\bfrac{1}{n^{1/2}}\brac{e^{o(n^{-1/8}a)}6^a(1-a)^2\brac{1+\frac{a}{2(1-a)}}^{2a}}^n.
\end{align*}
Now if $R_4(a)=\log\brac{6^a(1-a)^2\brac{1+\frac{a}{2(1-a)}}^{2a}}$ then
\begin{align*}
&R_4'(a)=\log 6-\frac{4}{2-a}+2\log\bfrac{2-a}{2-2a}\\
&R_4''(a)=\frac{2a}{(2-a)^2(1-a)}>0.
\end{align*}
Thus $R_4$ is log-convex on $[0,\frac{1}{2}]$. We have $R_4(0)=1$ and $R_4'(0)=\log 6 -2\leq -1/5$
and $R_4(2.01/4)<-1/20$. It follows from this and \eqref{43} that
\begin{multline}\label{B4}
B_{\ref{B4}}=\sum_{\ell<k=2}^{2.01n/4}\sum_{D=2k}^{k\log n}
\p_R(k,\ell,D)\leq\\
\sum_{\ell<k=2}^{n^{1/10}}\sum_{D=2k}^{4k}O\bfrac{1}{n^{1/2}}e^{o(n^{-1/8}a)}
+\sum_{\ell<k=n^{1/10}}^{2.01n/4}\sum_{D=2k}^{4k}O\bfrac{1}{n^{1/2}}e^{-(1/5+o(n^{1/8}))k}=o(1).
\end{multline}
{\bf Case 2.2c: $d=3$.}
\begin{align*}
[u^D]((1+u)^3-1-3u)^k &= [u^{D-2k}](3+u)^k\\
&= 3^{3k-D}\binom{k}{D-2k}.
\end{align*}

\begin{align*}
\p_R(k,\ell,D)&\leq O\bfrac{1}{n^{1/2}}\bfrac{e^{o(n^{-1/8}a)}a^{a\th}(1-a)^{3-a\th}}{h(a)^2}^n\bfrac{3^{3-\th}}{h\brac{\th-2}}^k\\
&=O\bfrac{1}{n^{1/2}}\brac{e^{o(n^{-1/8}a)}a^{a(\th-2)}(1-a)^{1-a(\th-2)}\bfrac{3^{3-\th}}{h(\th-2)}^a}^n\\
&= O\bfrac{1}{n^{1/2}}\brac{e^{o(n^{-1/8}a)}3^a(1-a)\bfrac{\bfrac{a}{3(1-a)}^{\th-2}}{h(\th-2)}^a}^n\\
&\le O\bfrac{1}{n^{1/2}}\brac{e^{o(n^{-1/8}a)}3^a(1-a)\brac{1+\frac{a}{3(1-a)}}^{a}}^n.
\end{align*}
Now if $R_5(a)=\log\brac{3^a(1-a)\brac{1+\frac{a}{3(1-a)}}^{a}}$ then
\begin{align*}
&R_5'(a)=-\frac{3}{3-2a}+\log\brac{\frac{3-2a}{1-a}}\\
&R_5''(a)=\frac{4a-3}{(3-2a)^2(1-a)}.
\end{align*}
{\bf Case 2.2c(i):} $.51\leq a\leq 1$.\\
Thus $R_5$ is log-concave on $[\frac{1}{2},\frac{3}{4}]$ and log-convex on $[\frac{3}{4},1]$. 
We have $R_5(1/2)=0$, $R_5'(1/2)\leq -1/10$
and $R_5(3/4)\leq -.04$ and $R_5(1)=0$ and $R_5'(1)=\infty$. It follows that
\begin{multline}\label{B4x}
B_{\ref{B4x}}=\sum_{\ell<k=.51n}^{n-n^{7/8}}\sum_{D=2k}^{3k}\p_R(k,\ell,D)\leq\\
\sum_{\ell<k=.51n}^{3n/4}\sum_{D=2k}^{3k}O\bfrac{1}{n^{1/2}}e^{-(k-n/2)/10+o(n^{-1/8}k)}
+\sum_{\ell<k=3n/4}^{n-n^{7/8}}\sum_{D=2k}^{3k}e^{-n^{7/8}}=o(1).
\end{multline}
Now let us consider $0\leq k\leq .51n$. 

{\bf Case 2.2c(ii):} $0\leq a\leq .51$.\\
(a) $\th\geq 2.0005$ and $0\leq a\leq $.\\
We go back to \eqref{3ayx} and make the choice $\z_1=\z_2=z$ and replace
$h(\th/d)^{-ad}$ by $\bfrac{3^{3-\th}}{h\brac{\th-2}}^a$ and 
consider the function 
$$F_1(\th,a)=\frac{h(\th a/3)^33^{(3-\th)a}}{h(a)^2h(\th-2)^a}$$
so that $\p_R(k,\ell,D)\leq O\bfrac{1}{n^{1/2}}F_1(\th,a)^n$. Let $G_1(\th,a)=\log(F_1(\th,a))$. Then
\begin{align}
&\frac{\partial G_1}{\partial a}=\log\bfrac{27(1-a)^2(\th-2)^{\th-2}(a\th)^\th}
{3^\th a^2(3-\th)^{3-\th}(3-a \th)^\th}\label{da}\\
&\frac{\partial G_1}{\partial\th}=a\brac{\log\bfrac{3-\th}{9}-\log(\th-2)+\log(a\th)-\log\brac{1-\frac{a\th}{3}}}\label{mmm}\\
&\frac{\partial^2G_1}{\partial a^2}=\frac{(3-a)\th-6}{a(1-a)(3-a\th)}\label{daa}\\
&\frac{\partial^2G_1}{\partial\th^2}=
\frac{((3-a)\th^2-12\th+18)a}{\th(3-a\th)(\th-3)(\th-2)}.\label{partial}
\end{align}
It follows from \eqref{daa} that 
\beq{isconvex}
G_1(\th,a)\text{ is a convex function of }a\text{ for }0\leq a\leq a_\th=\frac{3\th-6}{\th},\,\text{for $\th$ fixed,\,$2\leq \th\leq 3$}
\eeq
and 
\beq{isconcave}
G_1(\th,a)\text{ is a concave function of }a\text{ for $a_\th\leq a\leq 1$},\,\text{for $\th$ fixed,\,$2\leq \th\leq 3$}. 
\eeq
It follows from \eqref{partial} that
\beq{isisconcave}
G_1(\th,a)\text{ is a concave function of $\th$ on $[2,3]$ for $a$ fixed, $0\leq a\leq 1$}. 
\eeq
A calculation shows that if $g_1(\th)=G_1(\th,a_\th)$ then
\begin{align}
&g_1'(\th)=\frac{3}{\th^2}\log\bfrac{144}{3^{\th^2}(3-\th)^2}\label{loong}\\
&g_1''(\th)=-\frac{6}{\th^3(3-\th)}\brac{-\th+2(3-\th)\log\bfrac{12}{3-\th}}.\label{long}
\end{align}
Furthermore, if $g_2(\th)=\frac{\partial G_1}{\partial a}\mid_{a=a_\th}$ then
\beq{short}
g_2(\th)=\log\bfrac{12}{3^\th(3-\th)}.
\eeq
(a) $2.0005\leq \th\leq 3$ and $0\leq a\leq e^{-10000}$.\\
For $a\leq e^{-10000}$ we have $\frac{\partial G_1}{\partial a}\leq \log 10-(\th-2)\log1/a\leq -2$. So,
\beq{smalla}
F_1(\th,a)\leq e^{-2a}\text{ for }0\leq a\leq e^{-10000},\,2.0005\leq \th\leq 3.
\eeq

(b) $2.46\leq \th\leq 3$ and $e^{-10000}\leq a\leq .51$.\\
Now $a_\th>.51$ for $\th\geq 2.46$ and so \eqref{isconvex} implies that $G_1(\th,a)\leq\max\set{G_1(\th,e^{-10000}),G_1(\th,.51)}$
for $2.46\leq \th\leq 3$ and $e^{-10000}\leq a\leq .51$.
Now \eqref{smalla} implies
that $G_1(2.46,e^{-10000})<-2e^{-10000}$ and \eqref{mmm} implies that $\frac{\partial G_1}{\partial\th}\mid_{\th=2.46,a=e^{-10000}}<0$ 
and so \eqref{isisconcave} implies that $G_1(\th,e^{-10000})\leq -2e^{-10000}$ for $2.46\leq \th\leq 3$. Also,
by direct calculation, we have
$G_1(2.46,.51)<-.002$ and $\frac{\partial G_1}{\partial\th}\mid_{\th=2.46,a=.51}<0$ and so
$G_1(\th,.51)\leq -.002$ for $2.46\leq \th\leq 3$. Thus,
$$F_1(\th,a)\leq e^{-2e^{-10000}}\text{ for }e^{-10000}\leq a\leq .51\text{ and }2.46\leq \th\leq 3.$$

(c) $2.0005\leq \th\leq 2.25$ and $e^{-1000}\leq a\leq .51$.\\
We take $\z_1=.6$ and $\z_2=2.1$ in \eqref{3ayx} and let
$$F_2(\th,a)=F_1(\th,a)\frac{z^3}{f(z)}\frac{f(\z_1)^a}{\z_1^{\th a}}
\frac{f(\z_2)^{1-a}}{\z_2^{3-\th a}}=F_1(\th,a)e^{\r_2+\s_2 a+\t_2a\th}$$
where 
$$e^{\r_2}=\frac{z^3f(\z_2)}{f(z)\z_2^3},\,e^{\s_2}=\frac{f(\z_1)}{f(\z_2)},\,e^{\t_2}=\frac{\z_2}{\z_1}.$$

Let $G_2(\th,a)=\log(F_2(\th,a))$. $\frac{\partial^2G_2}{\partial a^2}=\frac{\partial^2G_1}{\partial a^2}$
and $\frac{\partial^2G_2}{\partial \th^2}=\frac{\partial^2G_1}{\partial \th^2}$
and so \eqref{isconvex},\eqref{isconcave} and \eqref{isisconcave} hold with $G_1$ replaced by $G_2$. 
Putting $\g_2(\th)=G_2(\th,a_\th)$ we see that $\g_2''(\th)=g_1''(\th)-\frac{12\s_2}{\th^3}>0$,
using \eqref{long} ($\s_2<-3.127$).
Thus $\g_2$ is convex on $2.0005\leq \th\leq
2.25$. Furthermore $\g_2(2.0005),\g_2(2.25)<-.00003$ and so $\g_2(\th)<-.00003$ for $\th\in [2.0005,2.25]$
and therefore $G_2(\th,a)\leq -.00003a/a_\th<-.00003a$ when $0\leq a\leq a_\th$ and $\th\in [2.0005,2.25]$.
Next let $\f_2(\th)=\frac{\partial G_2}{\partial a}\mid_{a=a_\th}$. We have $\f_2(\th)=g_2(\th)+\s_2+\t_2\th<-.05$ 
for $2.0005\leq \th\leq 2.25$, using \eqref{short} ($\t_2<1.253$). So,
$G_2(\th,a)\leq \f_2(\th)-.05(a-a_\th)$ for $a\geq a_\th$ when $\th\in [2.0005,2.25]$. Thus 
$$F_2(\th,a)<e^{-.00003a}\text{
for $e^{-1000}\leq a\leq .51$ and }2.0005\leq \th\leq 2.25.$$

Now suppose that we repeat the idea of the previous paragraph, but this time we take $\z_1=1.4$ and $\z_2=3$ in \eqref{3ayx}
and use the same notation. Putting $\g_2(\th)=G_2(\th,a_\th)$ we see that $\g_2''(\th)=g_1''(\th)-\frac{12\s_2}{\th^3}>0$,
using \eqref{long} ($\s_2<-2.27$).
Thus $\g_2$ is convex on $2.25\leq \th\leq
2.46$. Furthermore $\g_2(2.25),\g_2(2.46)<-.05$ and so $\g_2(\th)<-.05$ for $\th\in [2.25,2.46]$
and therefore $G_2(\th,a)\leq -.05a/a_\th<-.05a$ when $0\leq a\leq a_\th$ and $\th\in [2.25,2.46]$.
Next let $\f_2(\th)=\frac{\partial G_2}{\partial a}\mid_{a=a_\th}$. We have $\f_2(\th)=g_2(\th)+\s_2+\t_2\th<-.2$ 
for $2.25\leq \th\leq 2.46$, using \eqref{short} ($\t_2<.763$). So,
$G_2(\th,a)\leq \f_2(\th)-.2(a-a_\th)$ for $a\geq a_\th$ when $\th\in [2.25,2.46]$. Thus 
$$F_2(\th,a)<e^{-.05a}\text{
for $e^{-1000}\leq a\leq .51$ and }2.25\leq \th\leq 2.46.$$

(d) $2\leq \th\leq 2.0005$ and $e^{-1000}\leq a\leq .51$.\\ 
For this we simplify our estimate of $\p_R(k,\ell,D)$ by removing some terms involving $\b$ from \eqref{3ayxw}.
\begin{align}
&\p_R(k,\ell,D)\leq\Pr(\exists A\subseteq R,B\subseteq L:\;|A|=k,|B|=k-1,N_\G(A)\subseteq B,d_B(b)\geq 2,b\in B)\leq \nonumber\\
&O(n^{1/2})\binom{m}{k}\binom{n}{k-1}\sum_{\substack{2\leq  d_a,a\in[m]\\2\leq x_b\leq d,b\in[k-1]\\
\sum_{a\in [k]}d_a=\sum_{b\in [k-1]}x_b=D}}\prod_{a=1}^k\frac{z^{d_a}}{d_a!f(z)}
\prod_{b=1}^{k-1}\binom{d}{x_b}\ D!\prod_{i=0}^{D-1}\frac{1}{dn-i}=\nonumber\\
&O\bfrac{k}{m^{1/2}}\binom{n}{k}\binom{m}{k}\frac{z^{D}D!}{f(z)^k}\frac{(dn-D)!}{(dn)!}
\brac{\sum_{\substack{2\leq  d_a,a\in[k]\\\sum_ad_a=D}}\prod_{a=1}^k\frac{1}{d_a!}}
\brac{\sum_{\substack{2\leq x_b\leq d,b\in[k-1]\\\sum_bx_b=D}}\prod_{b=1}^{k-1}\binom{d}{x_b}}=\nonumber\\
&O\bfrac{k}{m^{1/2}}\binom{n}{k}\binom{m}{k}\frac{z^D}{f(z)^k}\frac{1}{\binom{dn}{D}}
\brac{[u^D](e^u-1-u)^k}\brac{[u^D]((1+u)^d-(1+du))^k}\leq\nonumber\\
&O\bfrac{k}{m^{1/2}}\binom{n}{k}\binom{m}{k}\frac{z^D}{f(z)^k}\frac{1}{\binom{dn}{D}}\frac{f(\z_1)^k}{\z_1^D}
\binom{dk}{D}=\nonumber\\
&O\bfrac{k}{m^{1/2}}\binom{n}{k}\binom{m}{k}\frac{\binom{dk}{D}}{\binom{dn}{D}}
\bfrac{f(\z_1)}{f(z)}^k\bfrac{z}{\z_1}^D\nonumber\\
&=O\bfrac{1}{m^{1/2}}\brac{\frac{h(\th a/d)^d}{h(a)h(a/\b)^\b h(\th/d)^{ad}}\brac{\frac{f(\z_1)}{\z_1^\th} 
\frac{z^\th}{f(z)}}^a}^n\nonumber\\
&=O\bfrac{1}{m^{1/2}}\brac{e^{o(n^{-1/8}a)}\frac{h(\th a/d)^d}{h(a)^2 h(\th/d)^{ad}}\brac{\frac{f(\z_1)}{\z_1^\th} 
\frac{z^\th}{f(z)}}^a}^n.\label{3axxx}
\end{align}
Now let
$$F_3(\th,a)=\frac{h(\th a/3)^33^{(3-\th)a}}{h(a)^2h(\th-2)^a}
\brac{\frac{f(\z_1)}{\z_1^\th}\frac{z^\th}{f(z)}}^a.$$ 
We take $\z_1=.0001$ and then 
$$\frac{f(\z_1)}{\z_1^\th}\frac{z^\th}{f(z)}<.e^{-.86}$$ 
for $2\leq \th\leq 2.0005$. 
Keeping some slack, we define 
$$F_4(\th,a)=\frac{h(\th a/3)^33^{(3-\th)a}e^{-.85a}}{h(a)^2h(\th-2)^a}$$
and $G_4(\th,a)=\log(F(\th,a))$.
Now let $\g_4(\th)=G_4(\th,a_\th)$. We have $\g_4'(\th)=g_1'(\th)-\frac{5.1}{\th^2}$ 
and $\g_4''(\th)=g_1''(\th)+\frac{10.2}{\th^3}$ and we find from
\eqref{long} that $\g_4$ is concave on $2\leq \th\leq 2.0005$. Furthermore $\g_4(2)=0$ and 
using \eqref{loong} we see that $\g_4'(2)<-.8$ and so $g_1(\th)<-.8(\th-2)$ for $\th\in [2,2.0005]$.
So $G_4(\th,a)\leq -.8a(\th-2)/a_\th\leq -.8a(\th-2)$ for $0\leq a\leq a_\th$.
Next let $\f_4(\th)=\frac{\partial G_4}{\partial a}
\mid_{a=a_\th}=g_2(\th)-.85$. We see from \eqref{short} that $g_2(\th)<-.5$ for $2\leq \th\leq 2.0005$ and thus 
$G_4(\th,a)\leq -.5(a-a_\th)$ for $a\geq a_\th$ when $\th\in [2,2.0005]$. 
Replacing $e^{-.85}$ by $e^{-.86}$ in the definition of $F_4(\th,a)$ we get $F_4(\th,a)<e^{-(4(\th-2)a/5+a/100)}$
for $0\leq a\leq .51$ when
$2\leq \th\leq 2.0005$.
So, for some small constant $c>0$,
\beq{B4y}
B_{\ref{B4y}}=\sum_{\ell<k=2}^{.51n}\sum_{D=2k}^{3k}\p_R(k,\ell,D)\leq
\sum_{\ell<k=2}^{.51n}\sum_{D=2k}^{3k}e^{-ck}=o(1).
\eeq
\subsubsection{Finishing the case $m\sim n$}\label{combine}
We repeat our observation that the maximum degree $\D$ in $\G$ is $o(\log n)$ \whp. Therefore

{\bf Case 1:} $m\geq n$.
$$\Pr(\m(\G)<n)\leq o(1)+\begin{cases}
 A_{\ref{A0}}+A_{\ref{A1}}+A_{\ref{A2}}+A_{\ref{A3}}+A_{\ref{A4}}+B_{\ref{B0}}+B_{\ref{B1}}+B_{\ref{B2}}
                          &d\geq 6\\
A_{\ref{A0}}+A_{\ref{A1}}+A_{\ref{A2}}+A_{\ref{A3}}+A_{\ref{A4'}}+B_{\ref{B0}}+B_{\ref{B3}}&d=5\\
A_{\ref{A0}}+A_{\ref{A1}}+A_{\ref{A2}}+A_{\ref{A3}}+A_{\ref{A4'}}+B_{\ref{B0}}+B_{\ref{B4}}&d=4\\
A_{\ref{A0}}+A_{\ref{A1}}+A_{\ref{A2}}+A_{\ref{A33}}+A_{\ref{A333}}+A_{\ref{A3333}}+A_{\ref{bbx}}+A_{\ref{A4'}}+B_{\ref{B4x}}+B_{\ref{B4y}}&d=3
                         \end{cases}
$$
where the $o(1)$ term accounts for $\Pr(\D(\G)>\log n)$.
We use $B_{\ref{B0}}+B_{\ref{B1}}+B_{\ref{B2}}$ to account for witnesses $A\subseteq L,B$ with $|A|\geq n-n^{7/8}$. 
This is because if $A'=R\setminus B$ and $B'=L\setminus A$ then $|A'|=m-k+1$ and $|B'|=n-k$ and $N_\G(A')\subseteq B'$
and there will be a minimal witness $A'',B''$ with $A''\subseteq A'$.

{\bf Case 2:} $m\leq n$.
$$\Pr(\m(\G)<m)\leq o(1)+\begin{cases}
B_{\ref{B0}}+B_{\ref{B1}}+B_{\ref{B2}}+A_{\ref{A1}}+A_{\ref{A2}}&d\geq 6\\
B_{\ref{B0}}+B_{\ref{B3}}+A_{\ref{A1}}+A_{\ref{A2}}&d=5\\
B_{\ref{B00}}+B_{\ref{B4}}+A_{\ref{A1}}+A_{\ref{A2}}&d=4\\
B_{\ref{B4x}}+B_{\ref{B4y}}+A_{\ref{A1}}+A_{\ref{A2}}&d=3
                         \end{cases}
$$
We point out for use in the next section that our computations allow us 
to claim that we have
\beq{next}
\sum_{\substack{k=n^\c\\ \ell\leq\min\set{k-1,m/2}}}^{n-n^\c}\sum_{D=dk}^{k\log n}\p_L(k,\ell,D)=O(e^{-\Omega(n^\c)}).
\eeq

Our computations also allow us to claim that
\beq{next1}
\sum_{\substack{k=n^\c\\ \ell\leq\min\set{k-1,n/2}}}^{n-n^\c}\sum_{D=2k}^{dk}\p_R(k,\ell,D)=O(e^{-\Omega(n^\c)}).
\eeq

\subsection{The case $m\geq n+n^{4/5}$}\label{gre}
Let $\cG(n,m)$ denote the set of bipartite graphs with $|L|=n,|R|=m$ that are $d$-regular on $L$ and 
degree at least 2 on $L$. Here $n+n^{4/5}\leq m\leq dn/2$. In fact suppose first that $m\leq \xi dn$ where $\xi<1/2$ is a constant.
Suppose that $G(n,m)$ is chosen uniformly at random from $\cG(n,m)$.

If there is no matching from $L$ to $R$, then let a minimal witness $A,B$ be {\em small} if $|A|\leq n^{3/4}$ 
and {\em large} if $|A|\geq n-n^{3/4}$ and {\em medium} otherwise. 
\subsubsection{Small/Large Witnesses}\label{smallw}
We go back to \eqref{1aa}. We see that $f(\z_1)<f(z)$ implies that the term 
$\frac{z^d}{f(z)^{\frac{\b-a}{1-a}}}\frac{f(\z_1)^{\frac{\b-a}{1-a}}}{\z_1^{d-x}}
\bfrac{e\frac{a}{1-a}}{x}^x$ is maximised over $\b\geq 1$ when $\b=1$. Next let $H(\b)=\b\log h(a/\b)$ then $H'(\b)=\log(1-a/\b)$ and 
$H''(\b)=\frac{a}{\b(\b-a)}$. Thus $h(a/\b)^\b$ is log-convex in $\b$ and so 
\beq{hab}
h(a/\b)^\b\geq \exp\set{H(1)+H'(1)(\b-1)}=h(a)(1-a)^{\b-1.}
\eeq
Going back to \eqref{42} we see that now we have
\beq{71}
\p_L(k,\ell,D)\leq O\bfrac{1}{n^{1/2}}\bfrac{\r_d(a)}{(1-a)^{\b-1}}^n.
\eeq
By taking $\e_L(\b)$ in place of $\e_L(1)$ we can take $K=\b$ in \eqref{43} and plugging this into \eqref{71}
we see that
\beq{noo5}
\sum_{\ell<k=2}^{n^\c}\sum_{D=dk}^{k\log n}\p_L(k,\ell,D)\leq O\bfrac{1}{n^{1/2}}\bfrac{e^{-\b a}}{(1-a)^{\b-1}}^n=o(1).
\eeq
To deal with $k\geq n-n^{\c}$ we treat this as $k\leq n^{\c}$ in Section \ref{AinR}. Indeed, if there is such 
a witness $A,B$, let $A'=R\setminus B$ and $B'=L\setminus A$. Then $N_\G(A')\subseteq B'$ and $|B'|<|A'|$ and so we can 
find a witness $A'',B''$ with $A''\subseteq A',B''\subseteq B'$ and $|B''|\leq n^\c$.

We use \eqref{3axxx} for this calculation. Now
\begin{align}
&\frac{h(\th a/d)^d}{h(a)h(a/\b)^\b h(\th/d)^{a d}}=
\frac{\bfrac{\th a}{d}^{\th a}\brac{1-\frac{\th a}{d}}^{d-\th a}}{a^a(1-a)^{1-a}\bfrac{a}{\b}^a\brac{1-\frac{a}{\b}}^{\b-a}
\bfrac{\th}{d}^{\th a}\brac{1-\frac{\th}{d}}^{da-\th a}}=\nonumber\\
&a^{(\th-2) a}\exp\set{-(d-\th a)\sum_{k=1}^\infty \frac{\th^ka^k}{kd^k}+a-\sum_{k=2}^\infty \frac{a^k}{k(k-1)}+
\frac{a}{\b}-\sum_{k=2}^\infty \frac{a^k}{\b^{k-1}k(k-1)}+
(da-\th a)\sum_{k=1}^\infty \frac{\th^k}{kd^k}}\nonumber\\
&=a^{(\th-2) a}\exp\set{a\brac{1+\frac{1}{\b}-(d-\th)\log(1-\th/d)-\th}+O(a^2)}\label{11}
\end{align}
So from \eqref{3axxx} we can write
\beq{noo1}
\p_R(k,\ell,D)\leq O\bfrac{1}{m^{1/2}}\brac{\bfrac{az}{\z_1}^{\th-2}\exp\set{1+\frac{1}{\b}-(d-\th)\log(1-\th/d)-\th+O(a)}
\frac{f(\z_1)}{\z_1^2} \frac{z^2}{f(z)}}^{a n}.
\eeq
Now we claim that 
\beq{zeta1}
\z_1^{\th-2}\geq \frac12\text{ and that }f(x)x^{-2}\text{ is monotone increasing in }x.
\eeq

First notice that $f(x)x^{-2} = \sum_{i=2}^\infty\frac{x^{i-2}}{i!}$
which is clearly monotone increasing. Second note that 
$\z_1 =\z(\th)$ and since $\frac{d\z(x)}{dx}>0$ we have
\begin{align*}
\lim_{x\to\infty}\frac{d\z(x)}{dx} &= \lim_{x\to\infty} \frac{f(\z(x))^2}{(e^{\z(x)}-1)^2-\z(x)^2e^{\z(x)}}\\
&= \lim_{\z\to\infty}\frac{f(\z)^2}{(e^\z-1)^2-\z^2e^\z} = 1
\end{align*}
and since $\z(x)$ is concave 
we have $\frac{d\z(x)}{dx}\ge 1$. This, along with $\lim_{x\to 2^{-}}\z(x) = 0$, implies that $\z(x)\ge x-2$. We can then lower bound
$$
\z_1^{\th-2} = \z(\th)^{\th-2}\ge (\th-2)^{\th-2} \ge e^{-e^{-1}}\ge 0.69
$$

Using this we see from
\eqref{noo1} that if
$$\th\geq \th_0=2+\frac{4}{\log (1/az) }$$
then
$$\p_R(k,\ell,D)\leq O\bfrac{1}{m^{1/2}}e^{-k}.$$
In which case we have
\beq{ps1}
\sum_{\ell<k=2}^{n^{\c}}\sum_{\th\geq \th_0}\p_R(k,\ell,D)\leq O\bfrac{1}{m^{1/2}}e^{-k}=o(1).
\eeq
When $\th<\th_0$ we have $\th=2+o(1),f(\z_1)/\z_1^2=1/2+o(1)$. Therefore
\beq{bor}
\p_R(k,\ell,D)\leq O\bfrac{1}{m^{1/2}}\bfrac{z^2e^{-(d-2)\log(1-2/d)+o(1)}}{2f(z)}^k.
\eeq
Now for $d\geq 4$ we have
\beq{noo2}
\frac{z^2e^{-(d-2)\log(1-2/d)+o(1)}}{2f(z)}\leq \frac{9}{10}
\eeq
and so
\beq{noo3}
\sum_{k=1}^{n^{\c}}\sum_{\th\leq \th_0}\p_R(k,\ell,D)\leq O\bfrac{1}{m^{1/2}}\bfrac{9}{10}^k=o(1).
\eeq
When $d=3$, the expression on the LHS of \eqref{noo2} is at most 1.26. So in this case we go back to 
\eqref{3axxx} and replace $\frac{1}{h(\th/d)^{ad}}$ by $\bfrac{3^{\th-2}}{h\brac{\th-2}}^a=e^{o(a)}$. After this
\eqref{11} is replaced by
$$a^{(\th-2) a}\exp\set{a\brac{1+\frac{1}{\b}+\th\log(\th/d)-\th}+o(a)}.$$
And then \eqref{bor} is replaced by
$$\p_R(k,\ell,D)\leq O\bfrac{1}{m^{1/2}}\bfrac{z^2e^{-2\log(3/2)+o(1)}}{2f(z)}^k\leq O\bfrac{1}{m^{1/2}}\frac{1}{2^k}$$
and so
\beq{noo4}
\sum_{k=1}^{n^{\c}}\sum_{\th\leq \th_0}\p_R(k,\ell,D)\leq O\bfrac{1}{m^{1/2}}\frac{1}{2^k}=o(1).
\eeq
\subsubsection{Medium Witnesses}
Let $d_i(n,m)$ denote the number of $R$-vertices of degree $i\geq 2$ in $G(n,m)$ and let $D_i(n,m)=\E(d_i(n,m))$.

We define three events: 
\begin{align}
&\cA_1(n,m-1)=\set{G\in \cG(n,m-1):\;\exists i:|d_i(n,m-1)-D_i(n,m-1)|> 
n^{\go}/i^3, 2\leq i\leq \log^2n}
\label{bad1}\\
&\cA_2(n,m-1)=\set{G\in \cG(n,m-1):\;\exists i:d_i(n,m-1)\neq 0, 
 i> \log^2n}
\label{bad1a}\\
&\cB(n,m)=\set{G\in \cG(n,m):\;|d_2(n,m)-D_2(n,m)|> 2n^{\go}}\label{bad2}
\end{align}
We argue next that if $\cA(n,m)=\cA_1(n,m-1)\cup\cA_2(n,m)$ then
\beq{sud}
\Pr(\cA(n,m)\cup \cB(n,m))=e^{-\Omega(\log^2n)}.
\eeq

For any $t>0$ we have
$$\Pr(|d_i(n,m-1)-D_i(n,m-1)|> t)\leq O(n^{1/2})\Pr(Bin(n,q_i) > t)$$
where $q_i=\frac{z^i}{i!f(z)}$.

We will now use the following bounds (see for example \cite{AS})
\begin{eqnarray}
\Pr(|Bin(n,p)-np|\geq t)&\leq& 2e^{-t^2/n},\label{cher1}\\
\Pr(Bin(n,p)\geq \a np)&\leq &(e/\a)^{\a np}.\label{cher2}
\end{eqnarray}
If $i\leq \log^2n$ then we can use \eqref{cher1} with $t=n^{\go}/i^3$ to deal with $\cA_1(n,m)$ and also with $\cB(n,m)$.
If $i\geq \log^2n$ then $nq_i\leq e^{-\Omega(\log^2n)}$. We can therefore use \eqref{cher2} with $\a=1/nq_i$ to deal with 
$\cA_2(n,m)$. This concludes the proof of \eqref{sud}.

Now consider a set of pairs $X\subseteq \cG(n,m-1)\times \cG(n,m)$. We place $(G_1,G_2)$ into $X$ if $G_2$ is obtained from $G_1$
in the following manner: Choose a vertex $x\in R$ of degree at least four in $G_1$. 
Suppose that its neighbours are $y_i,i=1,2,\ldots,k$
in any order. 
To create $G_2$ we (i) replace $x$ by two vertices $x$ and $m$ and then (ii) 
let the neighbours of $x$ in $G_2$ be $y_1,y_2$ and let the neighbours of $m$ be $y_3,\ldots,y_k$.

For $G\in \cG^*(n,m-1)$ let
$$\p_1(G)=|\set{G_2:(G,G_2)\in X}|$$
and for $G\in \cG^*(n,m)$ let
$$\p_2(G)=|\set{G_1:(G_1,G)\in X}|.$$
We note that if
$$\S_1=\sum_{i\geq 4}\binom{i}{2}D_i(n,m-1)$$ then
\begin{itemize}
 \item $G\notin \cA(n,m-1)$ implies that $\card{\p_1(G)-\S_1}\leq O(n^{\go})$.
\item $\p_1(G)\leq \binom{m-1}{2}$ for all $G\in \cG(n,m-1)$.
\item $G\notin \cB(n,m)$ implies that $\card{\p_2(G)-D_2(n,m)}\leq n^{\go}$.
\item $\p_2(G)\leq m$ for all $G\in \cG(n,m)$.
\end{itemize}
We then note that
\ignore{
\begin{align*}
\S_1-O(n^{\go})\leq &\frac{|X|}{|\cG(n,m)|}\leq \S_1+O(n^{\go})+m^2e^{-\Omega(\log^2n)}=\S_1+O(n^{\go}).\\
D_2(n,m+1)-n^{\go}\leq 
&\frac{|X|}{|\cG(n,m+1)|}\leq D_2(n,m+1)+n^{\go}+me^{-\Omega(\log^2n)}.
\end{align*}
}
$$ (\S_1-O(n^{\go}))|\cG(n,m-1)|\leq |X|\leq (D_2(n,m)+n^{\go}+me^{-\Omega(\log^2n)})|\cG(n,m)|.$$
Now let $\cP,\cQ$ be properties such that if $(G_1,G_2)\in X$ and $G_2\in\cQ$ then $G_1\in\cP$. Let $(G_1,G_2)$
be chosen uniformly from $X$ and let $\Pr_X$ denote probabilities computed w.r.t. this choice. Then
$$\Pr_X(G_2\in\cQ)\leq \Pr_X(G_1\in\cP)\leq \frac{|\cP|(\S_1+O(n^{\go}))+m|\cA(n,m-1)|}{|X|}$$
and
$$\Pr_X(G_2\in\cQ)\geq \frac{(|\cQ|-|\cB(n,m)|)(D_2(n,m)-n^{\go})}{|X|}$$
So,
$$\frac{(|\cQ|-|\cB(n,m)|)(D_2(n,m)-n^{\go})}
{|\cG(n,m)|(D_2(n,m)+n^{\go}+me^{-\Omega(\log^2n)})}
\leq \frac{|\cP|(\S_1+O(n^{\go}))+m|\cA(n,m-1)|}{|\cG(n,m-1)|(\S_1-O(n^{\go}))}.$$
So,
$$
\frac{|\cQ|}{|\cG(n,m)|}\leq (1+O(n^{-2/5}))\frac{|\cP|}{|\cG(n,m-1)|}.
$$
So, if $\cP_j$ is a property of $\cG(n,j)$ for $j=n,n+1,\ldots,m$, 
\beq{x1}
\frac{|\cP_m|}{\cG(n,m)}\leq (1+O(n^{-2/5}))^{m-n}\frac{|\cP_n|}{\cG(n,n+n^{4/5})}.
\eeq
We use \eqref{x1} in the following way: First let $\cB_{j},\,n+n^{4/5}\leq j\leq m$ be the property that $G\in \cG(n,j)$ 
contains a minimal witness $A,B$ with $A\subseteq L,
n^\c\leq |A|\leq n/2$.
If $(G_1,G_2)\in X$ and $G_2\in\cB_{m+1}$ then $G_1\in\cB_m$. 
Indeed $A,B\cap[m]$ is a witness in $G_1$.
Applying \eqref{x1} and \eqref{next} we see that 
\whp\ $\cB_m$ fails to occur. 
Now let $\cB'_j$ be the property that $G\in \cG(n,j)$ contains a minimal
witness $A,B$ with $A\subseteq R,n^\c\leq |A|,
|B|< \min\{|A|-(j-n),n/2\}$. 
If $(G_1,G_2)\in X$ and $G_2$ has a 
witness $A,B$ with $A\subseteq L$ and $n/2< |A|\leq n-n^\c$ 
then $G_2\in\cB'_m$. Indeed $A'=R\setminus A,B'=L\setminus B$ is 
also a witness in $G_2$. Now if $G_2\in\cB'_m$ with a witness
$A',B'$ then $A'\cap[m],B'$ is a witness in $G_1$ and so contains 
a minimal witness $A'',B''$ where 
$|A''|>|B''|+m-n>n^\c$ i.e. $G_1\in \cB_{m-1}'$.
Applying \eqref{x1} and \eqref{next1} we see that 
\whp\ $\cB_m'$ fails to occur. 
This deals with medium witnesses.

It only remains to consider $m$ close to $dn/2$ i.e. where $\xi$ defined at the beginning of this section 
is close 1/2. Observe first that the number of edges incident
with vertices of degree greater than two is at most $3dn(1-2\xi)$.
If there are $d_i$ vertices of degree $i=2,\geq3$ then $d_2+d_3=m=\xi dn$ and $2d_2+3d_3\leq dn$ which implies that $d_2\geq dn(3\xi-1)$.
So the number of edges incident with vertices of degree greater than two is at most $dn-2dn(3\xi-1)$.

Now consider a witness $A,B$ where $|A|=\g n$. We must have
$\g dn\leq 2\g n+3dn(1-2\xi)$ which implies that $\g\leq \frac{3d(1-2\xi)}{d-2}$ which can be made arbitrarily small.
Now the estimate in \eqref{71} will suffice up to $k\leq \e_Ln$ and so we only need to make $\xi$ close enough to 1/2 so that
$\g<\e_L$ (which depends only on $d$ and not $\g$). 
\subsection{The case $m\leq n-n^{4/5}$}\label{le}
We once again consider medium witnesses separately from small or large witnesses.
\subsubsection{Small/Large Witnesses}\label{smallww}
We first go back to \eqref{3axxx} and deal with $\p_R(k,\ell,D)$ for $k\leq n^\c$ as we did in Section \ref{gre}. 
For $k\geq n-n^\c$ we deal with $\p_L(k,\ell,D)$ for $k\leq n^\c$. We will go back to \eqref{1aa} and write 
$$\p_L(k,\ell,D)=O\bfrac{1}{n^{1/2}}\bfrac{h(a)^{d-1}}{h(a/\b)^{\b}}^n
\brac{\frac{z^d}{f(z)}\frac{f(\z_1)}{\z_1^{d-x}}
\bfrac{e\frac{a}{1-a}}{x}^x\bfrac{f(z)}{f(\z_1)}^{\frac{1-\b}{1-a}}}^{n-k}$$
Now $x=\frac{d(m-n)+(D-dk)+1}{m-k+1}\geq 0$ implies that 
$$1-\b\leq \frac{D-dk+1}{n}=O\bfrac{k\log n}{n}\text{ and that }x=O\bfrac{k\log n}{n}.$$
Also, $\z_1=\z(d-x)$ implies that $f(\z_1)=f(z)(1-O(x))$. Therefore,
$$\bfrac{f(z)}{f(\z_1)}^{\frac{1-\b}{1-a}}=e^{O(a^2\log^2n)}.$$
Arguing as for \eqref{71} we get
$$\p_L(k,\ell,D)\leq O\bfrac{1}{n^{1/2}}\bfrac{\r(a)e^{O(a^2\log^2n)}}{(1-a)^{\b-1}}^n.$$
Taking $\r(a)\leq e^{-a}$ as in \eqref{A2} and noting that $b\leq 1$ here we get
Thus
\beq{noo7}
\sum_{k=1}^{n^\c}\sum_{D=dk}^{k\log n}\p_L(k,\ell,D)\leq \sum_{k=1}^{n^\c}\sum_{D=dk}^{k\log n}
O\bfrac{1}{n^{1/2}}\brac{e^{-a+O(a^2\log^2n)}}^n=o(1).
\eeq
\subsubsection{Medium Witnesses}
Now consider a set of pairs $Y\subseteq \cG(n,m)\times \cG(n+1,m)$. We place $(G_1,G_2)$ into $Y$ if $G_2$ is obtained from $G_1$
in the following manner:  Choose $0\leq k\leq n$. Replace edges $(\ell,y)$ by $(\ell+1,y)$ for all $\ell>k$ and all $y$.
Add vertex $k+1$ and $d$ edges $(k+1,y_j),j=1,2,\ldots,d$.

Note that if $(G_1,G_2)\in Y$ and $G_1$ has a matching of $R$ into $L$ then so does $G_2$.

For $G\in \cG(n,m)$ let now
$$\p_1(G)=|\set{G_2:(G,G_2)\in Y}|$$
and for $G\in \cG(n+1,m)$ let
$$\p_2(G)=|\set{G_1:(G_1,G)\in Y}|.$$
Let
$$\S_2=(n+1)\brac{1-\frac{z^2}{2f(z)}}^d$$ 
and for $G\in \cG(n+1,m)$ let
$$L_3(G)=\card{\set{v\in L:\;\text{all neighbours of $v$ have degree at least 3}}}.$$
Let 
$$\cC(n+1,m)=\set{G\in \cG(n+1,m):\;|L_3(G)-\S_2|\leq n^{3/5}}.$$
Let 

We note that
\begin{itemize}
\item $G\in \cG(n,m)$ implies that $\p_1(G)=(n+1)\binom{m}{d}$.
\item $G\notin \cC(n,m+1)$ implies that $\card{\p_2(G)-\S_2}\leq n^{3/5}$.
\item $\p_2(G)\leq n+1$ for all $G\in \cG(n+1,m)$.
\end{itemize}
We then note that

\begin{align*}
&\frac{|Y|}{|\cG(n,m)|}= (n+1)\binom{m}{d}.\\
\S_2-n^{3/5}\leq &\frac{|Y|}{|\cG(n+1,m)|}\leq \S_2+n^{3/5}+(n+1)e^{-\Omega(\log^2n)}.
\end{align*}
Now let $\cP,\cQ$ be properties such that if $(G_1,G_2)\in Y$ and $G_2\in\cQ$ then $G_1\in\cP$. Let $(G_1,G_2)$
be chosen uniformly from $Y$ and let $\Pr_Y$ denote probabilities computed with respect to this choice. Then
$$\Pr_Y(G_2\in\cQ)\leq \Pr_Y(G_1\in\cP)= \frac{|\cP|(n+1)\binom{m}{d}}{|Y|}$$
and
$$\Pr_Y(G_2\in\cQ)\geq \frac{(|\cQ|-|\cC(n+1,m)|)(\S_2-n^{3/5})}{|Y|}$$
Arguing as in Section \ref{gre} we see that if
$\cP_j$ is a property of $\cG(j,m)$ for $j=m,m+1,\ldots,n$, 
\beq{y1}
\frac{|\cP_m|}{\cG(n,m)}\leq (1+O(n^{-2/5}))^{n-m}\frac{|\cQ|}{\cG(m+n^{4/5},m)}.
\eeq
First let $\cB_j,m+n^{4/5}\leq j\leq n$ be 
the property that $G\in \cG(j,m)$ contains a minimal witness $A,B$ with $A\subseteq R,n^\c\leq |A|\leq m/2$.
If $(G_1,G_2)\in X$ and $G_2\in\cB_{n+1}$ then $G_1\in\cB_n$. 
Indeed $A,B\cap[n]$ is a witness in $G_1$.
Applying \eqref{y1} and \eqref{next1} we see that 
\whp\ $\cB_n$ fails to occur. 
Now let $\cB_j'$ be the property that $G\in \cG(j,m)$ contains a minimal
witness $A,B$ with $A\subseteq R,n^\c\leq |A|,|B|\leq \min\{|A|-(j-m),m/2\}$.
If $(G_1,G_2)\in X$ and $G_2$ has a 
witness $A,B$ with $A\subseteq R$ and $m/2< |A|\leq m-n^\c$ 
then $G_2\in\cB'_m$. Indeed $A'=L\setminus A,B'=R\setminus B$ is 
also a witness in $G_2$. Now if $G_2\in\cB'_m$ with a witness
$A',B'$ then $A'\cap[m],B'$ is a witness in $G_1$ and so contains 
a minimal witness $A'',B''$ where 
$|A''|>|B''|+n-m>n^\c$ i.e. $G_1\in \cB_{m-1}'$.
Applying \eqref{y1} and \eqref{next} we see that 
\whp\ $\cB_m'$ fails to occur. 
This deals with medium witnesses.

\end{document}